\documentclass[10pt,journal]{IEEEtran}

\usepackage{amsmath}
\usepackage{amssymb}
\usepackage{mathrsfs}
\usepackage{mathdots}
\usepackage{bbm}
\usepackage{mathtools}
\usepackage{bm}
\usepackage{amsfonts}
\usepackage{cases}
\usepackage{graphicx}
\usepackage{cite}
\usepackage{multirow}
\usepackage{color}
\usepackage{algorithm}
\usepackage{algpseudocode}
\usepackage{graphicx}
\usepackage{subfigure}
\usepackage{epstopdf}
\usepackage{lipsum}
\usepackage{amsthm}
\usepackage{url}
\usepackage{tikz}
\usetikzlibrary{positioning,shapes}
\graphicspath{{./}{Fig/}}

\usepackage{hyperref}

\DeclareMathOperator*{\argmax}{argmax}
\DeclareMathOperator*{\argmin}{argmin}

\algnewcommand\algorithmicswitch{\textbf{switch}}
\algnewcommand\algorithmiccase{\textbf{case}}
\algnewcommand\algorithmicassert{\texttt{assert}}
\algnewcommand\Assert[1]{\State \algorithmicassert(#1)}%

\algnewcommand{\Initialize}[1]{%
  \State \textbf{Initialize:}
  \Statex \hspace*{\algorithmicindent}\parbox[t]{.8\linewidth}{\raggedright #1}
}
\algnewcommand{\Inputs}[1]{%
  \State \textbf{Inputs:}
  \Statex \hspace*{\algorithmicindent}\parbox[t]{.8\linewidth}{\raggedright #1}
}

\algnewcommand{\Try}[1]{%
  \State \textbf{Try:}
  \Statex \hspace*{\algorithmicindent}\parbox[t]{.8\linewidth}{\raggedright #1}
}

\algdef{SE}[SWITCH]{Switch}{EndSwitch}[1]{\algorithmicswitch\ #1\ \algorithmicdo}{\algorithmicend\ \algorithmicswitch}%
\algdef{SE}[CASE]{Case}{EndCase}[1]{\algorithmiccase\ #1}{\algorithmicend\ \algorithmiccase}%
\algtext*{EndSwitch}%
\algtext*{EndCase}%

\newtheorem{definition}{Definition}
\newtheorem{lemma}{Lemma}

\newtheorem{proposition}{Proposition}
\newtheorem{theorem}{Theorem}
\newtheorem{remark}{Remark}
\newtheorem{example}{Example}

\newtheorem{problem}{Problem}

\theoremstyle{definition}
\newtheorem*{example*}{Example}

\begin{document}

\title{\huge Continuous-time control synthesis under nested signal temporal logic specifications}
\author{ Pian Yu, Xiao Tan, and Dimos V. Dimarogonas
\thanks{This work was supported in part by the Swedish Research Council (VR), the Swedish Foundation for Strategic Research (SSF), the Knut and Alice Wallenberg Foundation (KAW), the ERC CoG LEAFHOUND project, and the EU CANOPIES project.}
\thanks{Pian Yu is currently at the Department of Computer Science, University of Oxford, United Kingdom. She was at the KTH Royal Institute of Technology when this work was conducted. Xiao Tan and Dimos V. Dimarogonas are with School of Electrical Engineering and Computer Science, KTH Royal Institute of Technology, 10044 Stockholm, Sweden.
        {\tt\small pian.yu@cs.ox.ac.uk, xiaotan,dimos@kth.se}}
}

\maketitle
\thispagestyle{empty}

\begin{abstract}
In this work, we propose a novel approach for the continuous-time control synthesis of nonlinear systems under nested signal temporal logic (STL) specifications. While the majority of existing literature focuses on control synthesis for STL specifications without nested temporal operators, addressing nested temporal operators poses a notably more challenging scenario and requires new theoretical advancements. 
Our approach hinges on the concepts of signal temporal logic tree (sTLT) and control barrier function (CBF). Specifically, we detail the construction of an sTLT from a given STL formula and a continuous-time dynamical system, the sTLT semantics (i.e., satisfaction condition), and the equivalence or under-approximation relation between sTLT and STL. 
Leveraging the fact that the satisfaction condition of an sTLT is essentially keeping the state within certain sets during certain time intervals, it provides explicit guidelines for the CBF design. The resulting controller is obtained through the utilization of an online CBF-based program coupled with an event-triggered scheme for online updating the activation time interval of each CBF, with which the correctness of the system behavior can be established by construction. We demonstrate the efficacy of the proposed method for single-integrator and unicycle models under nested STL formulas.
\end{abstract}


\ifCLASSOPTIONcaptionsoff
  \newpage
\fi

\section{Introduction}
High level formal languages, originated from computer science for the specification and verification of computer programs \cite{ Baier2008}, have attracted increasing attention  to a wider audience over the last decades, ranging from biological networks \cite{yan2022interpretable, sanwal2017combining} to robotics \cite{ kress2009temporal,  belta2007symbolic }. Temporal logics, such as Linear Temporal Logic (LTL)  and Signal Temporal Logic (STL) \cite{ maler2004 }, provide a rigorous, mathematical language characterizing the expected behaviors of the systems. LTL focuses on the Boolean satisfaction of events over a discrete-time state series. As a comparison, STL allows for characterizing system properties over dense time, and thus more favorable for continuous-time dynamical systems, e.g., robotic and cyber-physical system applications \cite{bartocci2018specification, eddeland2017objective}.

Designing control strategies for systems to satisfy high level specifications is known as the control synthesis problem. For LTL specifications, the classic automaton-based control synthesis scheme has been well-studied \cite{tabuada2009verification, belta2017formal} for hybrid and discrete-time dynamical systems. In recent years, several different control synthesis schemes are proposed for STL specifications. One popular approach is to evaluate the satisfaction of the STL specification over the sampled time instants, encode it as a mixed-integer program (MIP), and then solve it in a model predictive control framework \cite{raman2014model, raman2015reactive, sadraddini2015robust}. However, the exponential computational complexity with respect to the number of integer variables makes this approach difficult to be applied to STL formulas with long time horizons even for small dimensional dynamical systems. To address the exponential
complexity of integer-based optimization, 
recent work proposes to smoothly approximate the robustness metric of STL, and then sequential quadratic programming \cite{gilpin2020smooth} or convex-concave programming \cite{takayama2023signal} is proposed to find a solution. In \cite{lee2021signal}, STL formulas are interpreted over stochastic processes and the STL synthesis is reformulated as a probabilistic inference problem. Nevertheless, all these results are restricted to discrete-time systems.

There are some endeavours in recent years on the continuous-time  control synthesis problem for STL specifications, including, to name a few, the control barrier function-based  \cite{lindemann2018, lindemann2019,buyukkocak2022control}, automaton-based \cite{lindemann2020,ho2022automaton}, heuristic-based \cite{ mehdipour2018spatial}, sampling-based \cite{ vasile2017}, and learning-based  \cite{ yan2021neural,kapoor2020, varnai2019} methods.         
Different from the discrete-time control synthesis methods, most of the aforementioned approaches only can handle  STL formulas with non-nested temporal operators (we will refer to these formulas as non-nested STL for simplicity in the following). To be more specific, the CBF-based method \cite{lindemann2018}  deals with a fragment of non-nested STL formulas and linear predicates. The recent work in \cite{buyukkocak2022control} considers a richer STL fragment and provides heuristics on the decomposition and the ordering of sub-tasks which are then used to construct CBFs.  In \cite{ho2022automaton}, the sampling-based automaton-guided control synthesis approach allows the consideration of nonlinear predicates, yet it is still restricted to non-nested STL formulas. In \cite{lindemann2020}, a fragment of signal interval temporal logic formulas is considered for the automaton-based control synthesis. Moreover, the timed abstraction of the dynamical system is needed, which is based on the
assumption of existing feedback control laws. The case of STL formulas with nested temporal operators is substantially more challenging and requires new theoretical advancements. To the best of our knowledge, the continuous-time control synthesis for STL specifications with nested temporal operators is still an open problem.
  
In this work, we aim to develop an efficient control synthesis approach for continuous-time dynamical systems under STL specifications with nested temporal operators, e.g., $\mathsf{G}_{[a_1, b_1]}\mathsf{F}_{[a_2, b_2]}\mu$, $\mathsf{F}_{[a_1, b_1]} \mathsf{G}_{[a_2, b_2]} \mu $, $\mathsf{F}_{[a_1, b_1]}(\mu_1\mathsf{U}_{[a_2, b_2]}(\mathsf{F}_{[a_3, b_3]}\mu_2))$. 
Compared to previous CBF-based control synthesis works  \cite{lindemann2018, lindemann2019,buyukkocak2022control}, we provide a tangible tool, coined as the \emph{signal  temporal logic tree (sTLT)}, that explicitly transforms the satisfaction of an STL formula to a series of set invariance conditions, which naturally guides the design of corresponding CBFs.
 The main contributions of this work are summarized as follows. 1) We introduce a notion of sTLT, detail its construction from a given STL formula, its semantics (i.e., satisfaction condition),  and establish equivalence or under-approximation relation between sTLT and STL. 2) We show how to design CBFs and online update their activation time intervals under the guidance of the sTLT. The control synthesis scheme is given by an online CBF-based program. 3) We deduce the correctness of the system behavior under certain assumptions.

The remainder of this paper is organized as follows.
In Sec. II, we give some technical preliminaries  and introduce the continuous-time control synthesis problem. In Sec. III,  the notion of sTLT is introduced as well as its semantics. Then, we derive the equavilence or under-approximation relation between an STL
formula and its constructed sTLT. Finally, we show how to design the CBFs, online update their activation time intervals, and the  overall control synthesis scheme. Case studies with single integrator and unicycle dynamics are presented in Sec. IV. The work is then concluded in Sec. V.

\section{Preliminaries and Problem Formulation}

\textbf{Notation.} Let $\mathbb{R}:=(-\infty, \infty)$, $\mathbb{R}_{\ge 0}:=[0, \infty)$, and $\mathbb{N}:=\{0,1,2,\ldots\}$. Denote $\mathbb{R}^n$ as the $n$ dimensional real vector space, $\mathbb{R}^{n\times m}$ as the $n\times m$ real matrix space. Throughout this paper, vectors are denoted in italics, $x\in \mathbb{R}^n$, and boldface $\bm{x}$ is used for continuous-time signals. Let $\|x\|$ and $\|A\|$ be the Euclidean norm of vector $x$ and matrix $A$. Given a set $S\subset \mathbb{R}^n$, $\overline{S}$ denotes its complement and $\partial S$ denotes its boundary. Given a point $x\in \mathbb{R}^n$ and a set $S\subset \mathbb{R}^n$, the distance function is defined as $\texttt{dist}(x, S):=\inf_{y\in S}\|x-y\|$. The signed distance function $\texttt{sdist}(x, S)$ is defined as 
\begin{equation*}
    \texttt{sdist}(x, S)=\begin{cases}
    - \texttt{dist}(x,\overline{S}), & \text{if } x\in S, \\
    \texttt{dist}(x,S), & \text{if } x\notin S.
\end{cases}
\end{equation*}

Consider a continuous-time dynamical system of the form
\begin{equation}\label{x0}
\Sigma: \dot{x}=f(x, u),
\end{equation}
where $x\in \mathbb{R}^n$ and $u\in U\subseteq \mathbb{R}^m$ are respectively the state and input of the system, the function $f: \mathbb{R}^n\times \mathbb{R}^m\to \mathbb{R}^n$ is locally Lipschitz continuous in $x$ and $u$.

Let $\mathcal{U}$ be the set of all measurable functions that take their values in $U$ and are defined on $\mathbb{R}_{\ge 0}$. A curve $\bm{x}: \mathbb{R}_{\ge 0} \to \mathbb{R}^n$ is said to be a trajectory of (\ref{x0}) if there exists an input signal $\bm{u}\in \mathcal{U}$ satisfying (\ref{x0}) for almost all $t\in \mathbb{R}_{\ge 0}$. We use $\bm{x}_{x_0}^{\bm{u}}(t)$ to denote the trajectory point reached at time $t$ under the input signal $\bm{u}$ from the initial state $x_0$.

\subsection{Signal temporal logic}

Signal temporal logic (STL) \cite{maler2004} is a predicate logic based on continuous-time signals. When $\bm{x}: \mathbb{R}_{\ge 0}\to \mathbb{R}^n$ is considered, the predicate $\mu$ at time $t$ is obtained after evaluation of a predicate function $g_\mu: \mathbb{R}^n\to \mathbb{R}$ as follows
\begin{equation*}
   \mu:=\left\{\begin{aligned}
   \top, & \quad \text{if } \quad g_\mu(\bm{x}(t))\ge 0 \\
   \bot, & \quad \text{if } \quad g_\mu(\bm{x}(t))<0.
   \end{aligned}\right.
\end{equation*}

In \cite{sadraddini2015robust}, it was shown that each STL formula has an equivalent
STL formula in positive normal form (PNF), \textit{i.e.,} negations only occur adjacent to predicates. The syntax of the PNF STL is given by
\begin{equation}\label{Def:PNF}
  \begin{aligned}
&\hspace{0cm}\varphi ::= \top \mid \mu \mid \neg \mu \mid \varphi_1 \wedge \varphi_2 \mid \varphi_1 \vee \varphi_2 \\
&\hspace{3cm}\mid  \varphi_1\mathsf{U}_{[a, b]} \varphi_2\mid  \mathsf{F}_{[a, b]} \varphi\mid  \mathsf{G}_{[a, b]} \varphi,
 \end{aligned}
\end{equation}
where $\varphi, \varphi_1, \varphi_2$ are STL formulas and $[a,b], 0\le a\le b< \infty,$ denotes a time interval. Here, $\wedge$ and $\vee$ are logic operators ``conjunction" and ``disjunction", $\mathsf{U}_{[a, b]}$, $\mathsf{F}_{[a, b]}$, and $\mathsf{G}_{[a, b]}$ are temporal operators ``until", ``eventually", and ``always", respectively.

\begin{definition}[STL semantics \cite{raman2015reactive}]\label{STLsemantics}
The validity of an STL formula $\varphi$ with respect to a continuous-time
signal $\bm{x}$ evaluated at time $t$, is defined inductively as follows:
  \begin{eqnarray*}
  (\bm{x}, t) \vDash \mu &\Leftrightarrow& g_\mu(\bm{x}(t))\ge 0, \\
   (\bm{x}, t) \vDash  \neg  \mu &\Leftrightarrow& \neg((\bm{x}, t) \vDash \mu), \\
   (\bm{x}, t) \vDash \varphi_1 \wedge \varphi_2 &\Leftrightarrow& (\bm{x}, t) \vDash \varphi_1 \wedge  (\bm{x}, t) \vDash \varphi_2, \\
   (\bm{x}, t) \vDash \varphi_1 \vee \varphi_2 &\Leftrightarrow& (\bm{x}, t) \vDash \varphi_1 \vee  (\bm{x}, t) \vDash \varphi_2, \\
   (\bm{x}, t) \vDash \varphi_1 \mathsf{U}_{[a, b]} \varphi_2 &\Leftrightarrow& \exists {t'}\in [t+a, t+b]  \ \text{s.t.} \\
   & &(\bm{x}, {t'}) \vDash \varphi_2 \wedge \\
   & &\forall {t''}\in [t, t'],(\bm{x}, {t''}) \vDash \varphi_1, \\
   (\bm{x}, t) \vDash \mathsf{F}_{[a, b]} \varphi &\Leftrightarrow& \exists {t'}\in [t+a, t+b]  \ \text{s.t.} \\
   & &(\bm{x}, {t'}) \vDash \varphi, \\
   (\bm{x}, t) \vDash \mathsf{G}_{[a, b]} \varphi &\Leftrightarrow& \forall {t'}\in [t+a, t+b]  \ \text{s.t.} \\
   & &(\bm{x}, {t'}) \vDash \varphi.
  \end{eqnarray*}
\end{definition}

\begin{definition}\label{Def:feasibility}
  Consider the dynamical system $\Sigma$ in (\ref{x0}) and the STL formula $\varphi$ in (\ref{Def:PNF}). We say \emph{$\varphi$ is satisfiable from the initial state $x_0$} if there exists a control signal $\bm{u}\in \mathcal{U}$ such that
  $({\bm{x}}_{x_0}^{\bm{u}}, 0) \vDash \varphi.$
\end{definition}

Given an STL formula $\varphi$, the set of initial states from which $\varphi$ is satisfiable is denoted by
\begin{equation}\label{initialsatisfiableset}
  \mathbb{S}_{\varphi}:=\{x_0\in \mathbb{R}^n| \text{$\varphi$ is satisfiable from $x_0$}\}.
\end{equation}
For simplicity, we will refer to $\mathbb{S}_{\varphi}$ as \emph{the satisfying set} for $\varphi$ in the following. Please be aware that the computation of the set $\mathbb{S}_{\varphi}$ is tailored to the dynamical system $\Sigma$ under consideration. Here we omit it for notation simplicity.

\subsection{Reachability operators}

In this section, we define two reachability operators.

\begin{definition}\label{Def:maxreachset}
   Consider the system (\ref{x0}), a set $\mathcal{S} \subseteq \mathbb{R}^n$, and a time interval $[a, b]$. The \emph{maximal reachable set} $\mathcal{R}^M(\mathcal{S}, [a, b])$ is defined as
    \begin{multline*}
    \mathcal{R}^M(\mathcal{S},[a, b])=\Big\{x_0\in \mathbb{R}^n: \exists \bm{u}\in \mathcal{U}, \exists t'\in [a, b], \\
      \text{s.t. }\;  {\bm{x}}_{x_0}^{\bm{u}}(t')\in \mathcal{S} \Big\}.
    \end{multline*}
\end{definition}

\begin{definition}\label{Def:minreachset}
   Consider the system (\ref{x0}), the set $\mathcal{S}\subseteq \mathbb{R}^n$, and a time interval $[a, b]$. The \emph{minimal reachable set} $\mathcal{R}^m(\mathcal{S},[a, b])$ is defined as
    \begin{multline*}
    \mathcal{R}^m(\mathcal{S},[a, b])=\Big\{x_0\in \mathbb{R}^n: \forall \bm{u}\in \mathcal{U}, \exists t'\in [a, b], \\
      \text{s.t. }\; {\bm{x}}_{x_0}^{\bm{u}}(t')\in \mathcal{S} \Big\}.
    \end{multline*}
\end{definition}

The set $\mathcal{R}^M(\mathcal{S},[a, b])$ collects all states in $\mathbb{R}^n$ from which there exists an input signal $\bm{u}\in \mathcal{U}$ that drives the system to target set $\mathcal{S}$ at some time instant $t'\in [a, b]$. The set $\mathcal{R}^m(\mathcal{S},[a, b])$ collects all states in $\mathbb{R}^n$ from which no matter what input signal $\bm{u}\in \mathcal{U}$ is applied, the system can reach the target set $\mathcal{S}$ at some time instant $t'\in [a, b]$.

Let $\mathcal{S}$ be represented by the zero superlevel set of a continuous function: $\mathcal{S}=\{x\in \mathbb{R}^n: h_{\mathcal{S}}(x)\ge 0\}$. Similarly, let $\mathcal{R}^M(\mathcal{S},[a, b])$ and $\mathcal{R}^m(\mathcal{S},[a, b])$ be represented by the zero superlevel set of some continuous functions, i.e.,  	\begin{equation*}
\begin{aligned}
\mathcal{R}^M(\mathcal{S},[a, b])&:=\{x: h_{\mathcal{R}^M(\mathcal{S},[a, b])}(x)\ge 0\},\\
\mathcal{R}^m(\mathcal{S},[a, b])&:=\{x: h_{\mathcal{R}^m(\mathcal{S},[a, b])}(x)\ge 0\}.		\end{aligned}
\end{equation*} 
As shown in \cite{chen2018}, the calculation of maximal and minimal reachable sets can be casted as an optimal control problem, given below: 
	\begin{equation*}
		\begin{aligned}
			h_{\mathcal{R}^M(\mathcal{S},[a, b])}&(x)=\max_{\bm{u}\in \mathcal{U}}\max_{s\in [a, b]}h_{\mathcal{S}}(\bm{x}_{x}^{\bm{u}}(s)),\\
			h_{\mathcal{R}^m(\mathcal{S},[a, b])}&(x)=\min_{\bm{u}\in \mathcal{U}}\max_{s\in [a, b]}h_{\mathcal{S}}(\bm{x}_{x}^{\bm{u}}(s)).
		\end{aligned}
	\end{equation*}

In the following, relations are established between the STL temporal operators $\mathsf{F}_{[a, b]}$ and $\mathsf{G}_{[a, b]}$ and the maximal/minimal reachable sets. 

\begin{lemma}[\hspace{-0.4mm}\cite{chen2018}]\label{lem2}
	Given the system (\ref{x0}) and the STL predicate $\mu_1$, one has
	\begin{itemize}
		\item[i)] $\mathbb{S}_{\mathsf{F}_{[a, b]}\mu_1}= \mathcal{R}^M(\mathbb{S}_{\mu_1}, [a, b])$, and
		\item[ii)] $\mathbb{S}_{\mathsf{G}_{[a, b]} \mu_1}=  \overline{\mathcal{R}^m(\overline{\mathbb{S}_{\mu_1}}, [a, b])}$,
	\end{itemize}
where $\mathbb{S}_{\mathsf{F}_{[a, b]}\mu_1}$ and $\mathbb{S}_{\mathsf{G}_{[a, b]}\mu_1}$ are the satisfying sets for $\mathsf{F}_{[a, b]}\mu_1$ and $\mathsf{G}_{[a, b]}\mu_1$, respectively.
\end{lemma}

\begin{definition}
Given any STL formula $\varphi$, let its satisfying set $\mathbb{S}_\varphi$ in \eqref{initialsatisfiableset} be represented by the zero superlevel set of a function $h_{\mathbb{S}_\varphi}: \mathbb{R}^n\to \mathbb{R}$, i.e., $\mathbb{S}_\varphi=\{x\in \mathbb{R}^n: h_{\mathbb{S}_\varphi}(x)\ge 0\}$. We refer to such function $h_{\mathbb{S}_\varphi}$ as a \emph{value function} associated with the STL formula $\varphi$. 
\end{definition}
Note that given a predicate $\mu$,  $g_\mu$ is a value function associated with $\mu$. In the general case, we can use the signed distance function to denote a value function, i.e., $h_{\mathbb{S}_\varphi}(x) = -  \texttt{sdist}(x, {\mathbb{S}_\varphi})$.

\subsection{Time-varying control barrier functions}

Define a differentiable function $\mathfrak{b}: X \times [t_0, t_1] \to \mathbb{R}$ and the associated set
\begin{equation}\label{ct}
  \mathcal{C}(t):=\{x\in X | \mathfrak{b}(x, t)\ge 0\}.
\end{equation}
Then, we have the following definition.

\begin{definition} [CBF \cite{ames2016control}] \label{def:cbf}
A function $\mathfrak{b}: X \times [t_0, t_1] \to \mathbb{R}$ is called a \emph{valid control barrier function (vCBF)} for (\ref{x0}) if there exists a locally Lipschitz continuous class
$\mathcal{K}$ function $\alpha$ such that, for all $(x, t)\in \mathcal{C}(t)\times [t_0, t_1]$,
  \begin{equation}\label{eq:cbf_condition}
    \sup_{u\in U} \left\{\frac{\partial \mathfrak{b}(x, t)}{\partial x}f(x, u) +\frac{\partial \mathfrak{b}(x, t)}{\partial t}\right\} \ge -\alpha(\mathfrak{b}(x, t)).
  \end{equation}
\end{definition}
If $x_0 \in \mathcal{C}(t_0)$ and ${\mathfrak{b}(x, t)}$ is a vCBF, then any locally Lipschitz control $\bm{u}$ satisfying \eqref{eq:cbf_condition} guarantees $\bm{x}_{x_0}^{\bm{u}}(t)\in \mathcal{C}(t)$ for all $t\in [t_0,t_1]$. This can be shown by, for example, invoking the Comparison Lemma \cite{khalil1996nonlinear}.

\subsection{Problem formulation}

Before moving on, we first introduce the notion of nested STL formulas.

\begin{definition}[Nested STL formula]\label{nestedSTL}
We call an STL formula $\varphi$ \emph{nested} if it can be written in one of the following forms:
\begin{eqnarray}
\varphi=\mathsf{F}_{[a, b]}\varphi_1, \label{F}\\
\varphi=\mathsf{G}_{[a, b]}\varphi_1,\label{G}\\
\varphi=\varphi_1\mathsf{U}_{[a, b]}\varphi_2, \label{U}
\end{eqnarray}
where $\varphi_1$ in (\ref{F})-(\ref{G}) and at least one of $\varphi_1$ and $\varphi_2$ in (\ref{U}) include temporal operators. In addition, $\varphi_1$ in (\ref{F})-(\ref{G}) and  $\varphi_1, \varphi_2$ in (\ref{U}) are called the \emph{argument(s)}  of the STL formula $\varphi$.
\end{definition}

Examples of nested STL formulas include $\mathsf{F}_{[a_1, b_1]}\mathsf{G}_{[a_2, b_2]}\mu, \mathsf{G}_{[a_1, b_1]}\mathsf{F}_{[a_2, b_2]}\mu$, and $\mu_1\mathsf{U}_{[a_1, b_1]}(\mathsf{G}_{[a_2, b_2]}\mu_2\wedge \mathsf{F}_{[a_3, b_3]}\mu_3)$, etc.

In \cite{lindemann2018,lindemann2019}, continuous-time control-affine system of the form 
\begin{equation}\label{control-affine}
    \dot x=f(x)+g(x)u
\end{equation}
is considered, and appropriate CBFs are designed for a fragment of non-nested STL formulas, which we briefly recap here:
\begin{itemize}
  \item For $\mathsf{G}_{[a, b]}\mu_1$, select $\mathfrak{b}(x, t)$ s.t. $\mathfrak{b}(x, t')\le g_{\mu_1}(x)$ for all $t'\in [a, b]$,
  \item For $\mathsf{F}_{[a, b]}\mu_1$, select $\mathfrak{b}(x, t)$ s.t. $\mathfrak{b}(x, t')\le g_{\mu_1}(x)$ for some $t'\in [a, b]$.
  \item For $\mu_1\mathsf{U}_{[a, b]}\mu_2$, it is encoded as $\mathsf{G}_{[0, b]}\mu_1\wedge \mathsf{F}_{[a, b]}\mu_2$,
\end{itemize}
where $g_{\mu_1}$ and $g_{\mu_2}$ are the predicate functions of $\mu_1$ and $\mu_2$, respectively. Once an vCBF is obtained (the explicit CBF construction is investigated in \cite{lindemann2019}) for the non-nested STL formula, then the control strategy for (\ref{control-affine}) is given by solving a quadratic program (QP)
\begin{eqnarray} 
  &&\hspace{-0.2cm}\min_{u\in U} \quad u^TQu \nonumber\\
  &&\hspace{-1cm}\text{s.t.} \; \frac{\partial b(x, t)}{\partial x}(f(x)+g(x)u) +\frac{\partial b(x, t)}{\partial t} \ge -\alpha(b(x, t)).\label{constraint}
\end{eqnarray}

In this work, we consider the continuous-time control synthesis for nested STL formulas as per Definition \ref{nestedSTL}. Formally, the problem is stated as follows.

\begin{problem}\label{problem}
  Consider the dynamical system in (\ref{x0}) and a nested STL formula $\varphi$. Derive a continuous-time control strategy $\bm{u}$ such that the resulting trajectory $\bm{x}$ of (\ref{x0}) with initial state $x_0$ satisfies $\varphi$, i.e.,
  $(\bm{x}_{x_0}^{\bm{u}}, 0) \models \varphi$.
\end{problem}

\section{Solving the Control Synthesis Problem}

In this work, we aim to formulate the continuous-time control synthesis problem for a nested STL specification $\varphi$ as a CBF-based program as in \cite{lindemann2018}. Here, the difficulty is: how to encode the task satisfaction constraint (i.e., $(\bm{x}_{x_0}^{\bm{u}}, 0) \models \varphi$) as a set of constraints on the system input $u$ when $\varphi$ is nested? Appropriate CBFs have been  proposed in \cite{lindemann2018} for control-affine systems under non-nested  STL formula $\varphi$, e.g., $\varphi=\mathsf{F}_{[a, b]}\mu$. However, when the STL formula $\varphi$ is nested, extending the CBF design methodology in \cite{lindemann2018} to nested STL formulas is nontrivial. 

To tackle this problem, in this work, we propose the notion of sTLT. This tree structure serves as a tool for guiding the design of CBFs for nested STL formulas.

This section is structured as follows. First, we introduce the notion of sTLT and its construction in Section III. A. Then we define sTLT semantics in Section III.B. The equivalence or under-approximation relation between STL and sTLT is discussed in Section III. C. Then, we explain the design of the CBFs based on the sTLT in Section III. D. In Section III. E, we show the overall algorithm. Finally in Section III. F, the computational complexity of the overall approach is discussed.

\subsection{sTLT and its construction} \label{subsec:treeconstruction}

An sTLT refers to a tree with linked set nodes and operator nodes. The formal definition is given as follows.

\begin{definition}[sTLT]
An \emph{sTLT} is a tree for which the next holds:
\begin{itemize} 
\item each node is either a \emph{set} node that is a subset of $\mathbb{R}^n$ or an \emph{operator} node that belongs to $\{\wedge, \vee, \mathsf{U}_{[a, b]},\mathsf{F}_{[a, b]},\mathsf{G}_{[a, b]}\}$;
\item the root node and the leaf nodes are \emph{set} nodes;
\item if a \emph{set} node is not a leaf node, its unique child is an \emph{operator} node;
\item the children of any \emph{operator} node are \emph{set} nodes.
  \end{itemize}
\end{definition}

The sTLT is motivated by the notion of TLT defined in \cite{gao2021} for LTL formulas. Although graphically similar, the sTLT construction and its satisfaction condition are substantially
different from TLT in \cite{gao2021}. We will provide additional clarification regarding the differences later in Remark \ref{Rem1}. 

\subsubsection*{Construct an sTLT from an STL formula $\varphi$} Before presenting the construction procedure of such an sTLT from a given STL formula $\varphi$ and a continuous-time dynamical system $\Sigma$, we give the following definition.

\begin{definition}[Desired form]\label{def:desired}
   Given an STL formula $\varphi$ in Definition \ref{Def:PNF}, we say \emph{$\varphi$ is in desired form} if i) it contains no ``until" operators and ii) the argument of every ``always" operator contains no ``disjunction" operator.
\end{definition}

\begin{figure*}[t]
\centering
\subfigure[]{
\includegraphics[width=0.15\textwidth]{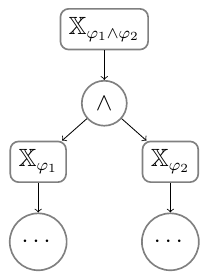}}
	\label{Fig:wedge}
\subfigure[]{	\includegraphics[width=0.15\textwidth]{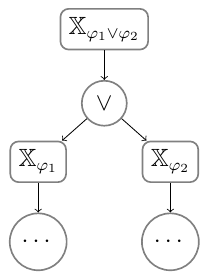}}
	\label{Fig:vee}
\subfigure[]{	\includegraphics[width=0.25\textwidth]{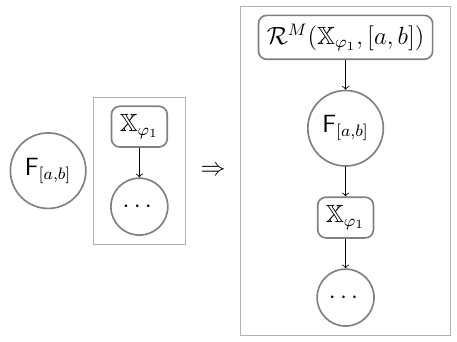}}
	\label{Fig:eventually}
\subfigure[]{	\includegraphics[width=0.25\textwidth]{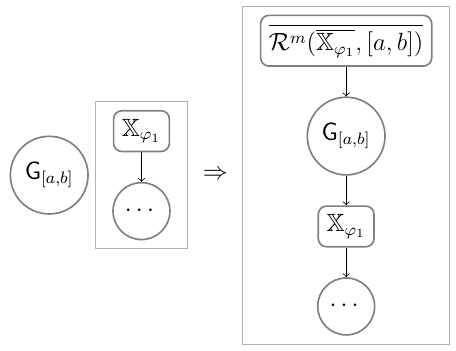}}
	\label{Fig:always}
\caption{The sTLT construction for: (a) $\varphi_1\wedge \varphi_2$; (b) $\varphi_1\vee \varphi_2$; (c) $\mathsf{F}_{[a, b]}\varphi_1$; (d) $\mathsf{G}_{[a, b]}\varphi_1$. The circles denote the operator nodes and the rectangles denote the set nodes.} \vspace{-0.2cm}
\end{figure*}

Next we detail the construction of sTLT from an STL formula $\varphi$ using the reachability operators $\mathcal{R}^M$ and $\mathcal{R}^m$, 
which can be completed in 3 steps.

\emph{Step 1:} Rewrite the STL formula $\varphi$ into the desired form $\hat\varphi$ as per Definition \ref{def:desired}. That is, i) if $\varphi$ contains ``until" operator, e.g., $\varphi=\varphi_1\mathsf{U}_{[a, b]} \varphi_2$, it is encoded as $\hat\varphi=\mathsf{G}_{[0, b]}\varphi_1 \wedge \mathsf{F}_{[a, b]}\varphi_2$ and ii) if the argument of a temporal operator contains a ``disjunction" operator, e.g., $\varphi=\Theta_{[a, b]}(\varphi_1\vee \varphi_2)$, it is encoded as $\hat\varphi=\Theta_{[a, b]}\varphi_1\vee \Theta_{[a, b]}\varphi_2$, where $\Theta \in \{\mathsf{G}, \mathsf{F}\}$. After this step, one has that $\hat\varphi$ contains no ``until" operator and the ``disjunction" operator, if it exists, appears in the form of $\hat{\varphi} = \varphi_1 \vee \varphi_2 \vee \ldots \vee \varphi_N$, and $\varphi_i, i= 1, 2, \ldots, N,$ contain no ``disjunction" operator. We call the fragment of STL formulas $\hat\varphi$, identified by Definition \ref{def:desired}, \emph{desired form}. This is because we will later observe that the constructed sTLT $\mathcal{T}_{\hat\varphi}$ is equivalent to $\hat\varphi$ in the sense that every trajectory that satisfies the sTLT $\mathcal{T}_{\hat\varphi}$ also satisfies the STL formula $\hat\varphi$, and conversely.

\emph{Step 2:} For each predicate $\mu$ or its negation $\neg \mu$, construct the sTLT with only a single set node $\mathbb{X}_{\mu}=\mathbb{S}_{\mu}=\{x: g_{\mu}(x)\ge 0\}$ or $\mathbb{X}_{\neg \mu}=\mathbb{S}_{\neg\mu}=\{x: -g_{\mu}(x)\ge 0\}$. The sTLT of $\top$ or $\bot$ has only a single set node, which is $\mathbb{R}^n$ or $\emptyset$, respectively.

\emph{Step 3:} Construct the sTLT $\mathcal{T}_{\hat\varphi}$ inductively. More specifically, for given STL formulas $\varphi_1$ and $\varphi_2$ and their corresponding constructed sTLTs $\mathcal{T}_{\varphi_1}, \mathcal{T}_{\varphi_2}$, the sTLT from a) $\varphi_1 \wedge \varphi_2$, b) $\varphi_1 \vee \varphi_2$, c) $\mathsf{F}_{[a, b]} \varphi_1$, and d) $\mathsf{G}_{[a, b]} \varphi_1$  can be constructed following the rules detailed below. Denote by $\mathbb{X}_{\varphi_1}:=\{x: h_{\mathbb{X}_{\varphi_1}}\ge 0\}$ and $\mathbb{X}_{\varphi_2}:=\{x: h_{\mathbb{X}_{\varphi_2}}\ge 0\}$ the root nodes of $\mathcal{T}_{\varphi_1}$ and $\mathcal{T}_{\varphi_2}$, respectively.

Case a): Boolean operator $\wedge$. The sTLT $\mathcal{T}_{\varphi_1 \wedge \varphi_2}$ can be constructed by connecting $\mathbb{X}_{\varphi_1}$ and $\mathbb{X}_{\varphi_2}$ through the operator node $\wedge$ and taking $$\mathbb{X}_{\varphi_1\wedge \varphi_2}:=\{x:  (h_{\mathbb{X}_{\varphi_1}}\ge 0) \wedge (h_{\mathbb{X}_{\varphi_2}}\ge 0)\}$$ to be the root node. An illustrative diagram for $\varphi_1 \wedge \varphi_2$ is given in Fig. 1(a).

Case b): Boolean operator $\vee$. The sTLT $\mathcal{T}_{\varphi_1 \vee \varphi_2}$ can be constructed by connecting $\mathbb{X}_{\varphi_1}$ and $\mathbb{X}_{\varphi_2}$ through the operator node $\vee$ and taking $$\mathbb{X}_{\varphi_1\vee \varphi_2}:=\{x:  (h_{\mathbb{X}_{\varphi_1}}\ge 0) \vee (h_{\mathbb{X}_{\varphi_2}}\ge 0)\}$$ to be the root node. An illustrative diagram for $\varphi_1 \vee \varphi_2$ is given in Fig. 1(b).

Case c): Eventually operator $\mathsf{F}_{[a, b]}$. The sTLT $\mathcal{T}_{\mathsf{F}_{[a, b]}{\varphi_1}}$ can be constructed by connecting  $\mathbb{X}_{\varphi_1}$ through the operator $\mathsf{F}_{[a, b]}$ and making the set $\mathcal{R}^M(\mathbb{X}_{\varphi_1}, \mathbb{R}^n, [a, b])$ the root node. An illustrative diagram for $\mathsf{F}_{[a, b]}\varphi_1$ is given in Fig. 1(c).

Case d): Always operator $\mathsf{G}_{[a, b]}$. 
The sTLT $\mathcal{T}_{\mathsf{G}_{[a, b]}{\varphi_1}}$ can be constructed by connecting  $\mathbb{X}_{\varphi_1}$ through the operator $\mathsf{G}_{[a, b]}$ and making the set $\overline{\mathcal{R}^m(\overline{\mathbb{X}_{\varphi_1}}, [a, b])}$ the root node. An illustrative diagram for $\mathsf{G}_{[a, b]}\varphi_1$ is given in Fig. 1(d).

Thus we complete the construction of an sTLT from a STL formula $\varphi$. In what follows, if not stated otherwise, we will use $\hat{\varphi}$ as the desired form of $\varphi$ obtained from Step 1 and $\mathcal{T}_{\hat{\varphi}}$ as the constructed sTLT for brevity. 

Let us use the following example to show how to construct an sTLT from a nested STL formula.      
                                      
\begin{example}\label{example1}
Consider the nested STL formula $\varphi=\mathsf{F}_{[0, 15]}(\mathsf{G}_{[2, 10]}\mu_1 \vee \mu_2  \mathsf{U}_{[5, 10]}\mu_3)$, where $\mu_i, i=\{1,2,3\}$ are predicates. Following Step 1, we can rewrite $\varphi$ into the desired form $\hat \varphi= \mathsf{F}_{[0, 15]}\mathsf{G}_{[2, 10]}\mu_1 \vee  \mathsf{F}_{[0, 15]}(\mathsf{G}_{[0, 10]}\mu_2  \wedge \mathsf{F}_{[5, 10]}\mu_3)$. The constructed sTLT $\mathcal{T}_{\hat\varphi}$ is plotted in Fig. \ref{Fig:example}. Recall that the sTLT is constructed in a bottom-up manner, i.e., we first construct the leaf nodes corresponding to the three predicates, i.e., $\mathbb{X}_5=\mathbb{S}_{\mu_1}$, $\mathbb{X}_8=\mathbb{S}_{\mu_2}$, $\mathbb{X}_9=\mathbb{S}_{\mu_3}$, and then build upon them one can compute 
\begin{eqnarray*}
	&& \mathbb{X}_3=\overline{\mathcal{R}^m(\overline{\mathbb{X}_5}, [2, 10])},\\
	&& \mathbb{X}_1=\mathcal{R}^M(\mathbb{X}_3, [0, 15]),\\
	&& \mathbb{X}_6=\overline{\mathcal{R}^m(\overline{\mathbb{X}_8}, [0, 10])},\\
	&& \mathbb{X}_7=\mathcal{R}^M(\mathbb{X}_9, [5, 10]),\\
	&& \mathbb{X}_4=\{x: h_{\mathbb{X}_{6}}(x)\ge 0 \wedge h_{\mathbb{X}_{7}}(x)\ge 0\},\\
	&& \mathbb{X}_2=\mathcal{R}^M(\mathbb{X}_{4}, [0, 15]),\\
	&& \mathbb{X}_0=\{x: h_{\mathbb{X}_{1}}(x)\ge 0 \vee h_{\mathbb{X}_{2}}(x)\ge 0\}.
\end{eqnarray*}
\end{example}

\begin{figure}[t]
\centering	\includegraphics[width=0.2\textwidth]{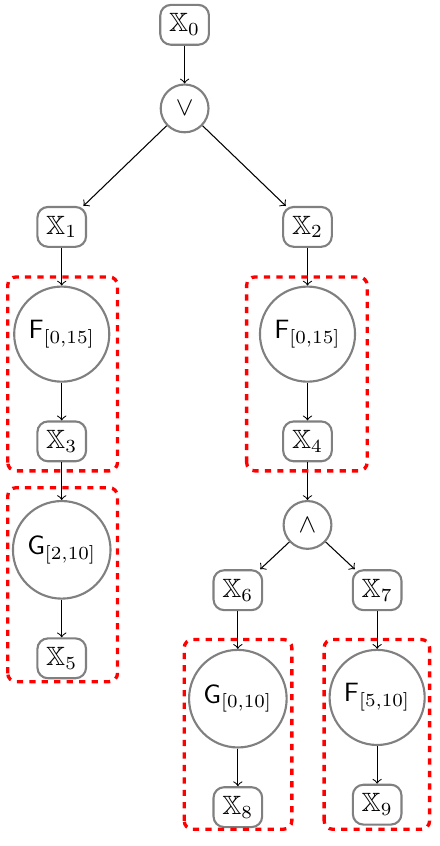}
	\caption{\footnotesize The sTLT $\mathcal{T}_{\hat\varphi}$ for the nested STL formula $\varphi=\mathsf{F}_{[0, 15]}(\mathsf{G}_{[2, 10]}\mu_1 \vee \mu_2  \mathsf{U}_{[5, 10]}\mu_3)$. }
 \label{Fig:example}
\end{figure}

\subsection{sTLT semantics}

Before define the \emph{sTLT semantics}, i.e., the satisfaction relation between a trajectory $\bm{x}$ and an sTLT $\mathcal{T}$, the definitions of complete path, temporal fragment, and time interval coding for an sTLT $\mathcal{T}$ are needed.

\begin{definition}[Complete path]\label{def:completepath}
 A \emph{complete path} $\bm{p}$ of an sTLT is a path that starts from the root node and ends at a leaf node. It can be encoded in the form of $\bm{p}=\mathbb{X}_0\Theta_1\mathbb{X}_1\Theta_2\ldots \Theta_{N_f} \mathbb{X}_{N_f}$, where $N_f$ is the number of operator nodes contained in the complete path, $\mathbb{X}_i, i\in \{0,1,\ldots,N_f\}$ represent set nodes, and $\Theta_{j}, \forall j\in \{1,\ldots, N_f\}$ represent operator nodes. 
\end{definition}

\begin{definition}[Temporal fragment]\label{def:temporalfragment}
A \emph{temporal fragment}  of a complete path is a fragment that starts from one temporal operator node, i.e., the node $\mathsf{U}_{[a, b]}, \mathsf{F}_{[a, b]}$ or $\mathsf{G}_{[a, b]}$, and ends at its child set node.
\end{definition}

\begin{definition}[Time interval coding]\label{Def:timecoding}
	A \emph{time interval coding} of a complete path involves assigning a time interval $[\underline{t}_{i}, \bar{t}_{i}],  0\le \underline{t}_{i}\le \bar{t}_{i}$ to  each set node $\mathbb{X}_i$ in the complete path.
\end{definition}

Given a time instant $\hat t$ and two time intervals $[a_1, b_1], [a_2, b_2]$, define
\begin{equation*}
	\begin{aligned}
		&\hat t+[a_1, b_1]:=[\hat t+a_1, \hat t+b_1],\\
		&[a_1, b_1]+[a_2, b_2]:=[a_1+a_2, b_1+b_2].
	\end{aligned}
\end{equation*}
Now, we further define the satisfaction relation between a trajectory $\bm{x}$ and a complete path of the sTLT.

\begin{definition}\label{Def:PathSaf}
	Consider a trajectory $\bm{x}$ and a complete path $\bm{p}=\mathbb{X}_0\Theta_1\mathbb{X}_1\Theta_2\ldots \Theta_{N_f} \mathbb{X}_{N_f}$. We say \emph{$\bm{x}$ satisfies $\bm{p}$ from time $t$}, denoted by $(\bm{x}, t)\cong \bm{p}$, if there exists a time interval coding for $\bm{p}$ such that $\underline{t}_{0}=\bar{t}_0=t$ and, for $i=1,2, \ldots, N_f $,
	\begin{itemize}
		\item[i)] if $\Theta_i\in \{\wedge, \vee\}$, then $[\underline{t}_{i}, \bar{t}_{i}]=[\underline{t}_{{i-1}}, \bar{t}_{{i-1}}]$;
		\item[ii)] if $\Theta_i\in \{\mathsf{U}_{[a, b]}, \mathsf{F}_{[a, b]}\}$, then $\exists t'\in [a, b]$ s.t. $[\underline{t}_{{i}}, \bar{t}_{{i}}]=t'+[\underline{t}_{{i-1}}, \bar{t}_{{i-1}}]$;
		\item[iii)] if $\Theta_i=\mathsf{G}_{[a, b]}$, then $[\underline{t}_{{i}}, \bar{t}_{{i}}]=[a, b]+[\underline{t}_{{i-1}}, \bar{t}_{{i-1}}]$;
	\end{itemize}
	and, for $i=0,1, \ldots, N_f $,
	\begin{itemize}
		\item[iv)] $\bm{x}(t)\in \mathbb{X}_{i}, \forall t\in [\underline{t}_{i}, \bar{t}_{i}]$.
	\end{itemize}
\end{definition}

With Definition \ref{Def:PathSaf}, the sTLT semantics, i.e., the satisfaction relation between a trajectory $\bm{x}$ and an sTLT, can be defined as follows.

\begin{definition}[sTLT semantics]\label{Def:TreeSaf}
	Consider a trajectory $\bm{x}$ and an sTLT $\mathcal{T}$. We say \emph{$\bm{x}$ satisfies $\mathcal{T}$ from time $t$}, denoted by $(\bm{x}, t) \cong \mathcal{T}$, if the output of Algorithm  \ref{alg:tree satisfaction} is ${\rm true}$.
\end{definition}

\begin{algorithm}
	\caption{\textit{sTLT Satisfaction}} \label{alg:tree satisfaction}
	\begin{algorithmic}[1]
\Require a trajectory $\bm{x}$ and an sTLT $\mathcal{T}$.
\Ensure ${\rm true}$ or ${\rm false}$.
\State $\mathcal{T}^c \leftarrow$ remove all temporal fragments in $\mathcal{T}$,
\For {each complete path $\bm{p}$ of $\mathcal{T}$,}
\If{$(\bm{x}, 0)\cong \bm{p}$}
\State \parbox[t]{\dimexpr\linewidth-\algorithmicindent}{set the corresponding leaf node of $\bm{p}$ \\in $\mathcal{T}^c$ to ${\rm true}$,}
\Else
\State \parbox[t]{\dimexpr\linewidth-\algorithmicindent}{set the corresponding leaf node of $\bm{p}$\\ in $\mathcal{T}^c$ to ${\rm false}$,}
\EndIf
\EndFor
\State set all the non-leaf set nodes of $\mathcal{T}^c$ to ${\rm false}$,
\State $\mathcal{T}^c\leftarrow \textit{Backtracking}(\mathcal{T}^c)$,
\State return the root node of $\mathcal{T}^c$.
\end{algorithmic}
\end{algorithm}

\begin{figure}[t]
\centering	\includegraphics[width=0.2\textwidth]{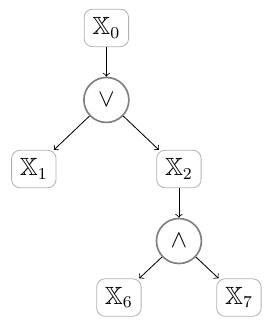}
	\caption{\footnotesize The compressed tree $\mathcal{T}^c$ for the sTLT $\mathcal{T}_{\hat\varphi}$ plotted in Fig. \ref{Fig:example}. }
 \label{Fig:example_compressed}.
\end{figure}

Algorithm \ref{alg:tree satisfaction} takes as inputs  a trajectory $\bm{x}$ and an sTLT $\mathcal{T}$. The output is ${\rm true}$ or ${\rm false}$. It works as follows. Given the sTLT $\mathcal{T}$, we first remove all its temporal fragments (line 1). When removing a temporal fragment, we reconnect the parent node of the corresponding temporal operator node and the child of the corresponding set node. In this way the resulting compressed tree $\mathcal{T}^c$ contains only Boolean operator nodes and set nodes. For the sTLT $\mathcal{T}_{\hat\varphi}$ shown in Fig. \ref{Fig:example}, the compressed tree $\mathcal{T}^c$ is depicted in Fig. \ref{Fig:example_compressed}. Then for each complete path $\bm{p}$ of $\mathcal{T}$, if $(\bm{x}, 0)\cong \bm{p}$, one sets the corresponding leaf node of $\bm{p}$ in $\mathcal{T}^c$ (note that $\mathcal{T}^c$ and $\mathcal{T}$ have the same number of leaf nodes) to ${\rm true}$. Otherwise, one sets the corresponding leaf node of $\bm{p}$ in $\mathcal{T}^c$ to ${\rm false}$ (lines 2-8). After that, we set all the non-leaf set nodes of $\mathcal{T}^c$ to ${\rm false}$ (line 9) and the resulting tree becomes a Boolean tree (a tree with Boolean operator and Boolean variable nodes). Finally, we backtrack the Boolean tree $\mathcal{T}^c$ using Algorithm \textit{Backtracking}, given in Algorithm  \ref{alg:Backtracking}, and return the root node (lines 10-11).

\begin{algorithm} 
	\caption{\textit{Backtracking}} \label{alg:Backtracking}
	\begin{algorithmic}[1]
		\Require a tree $\mathcal{T}^c$ with  Boolean operator and Boolean variable nodes.
		\Ensure the updated $\mathcal{T}^c$.
		\For {each operator node $\Theta$ of $\mathcal{T}^c$ through a bottom-up traversal,}
		\If {$\Theta=\wedge$,}
		\State $\text{PA}(\Theta)\leftarrow \text{PA}(\Theta)\vee (\text{CH}_1(\Theta)\wedge \text{CH}_2(\Theta))$,
		\Else
		\State $\text{PA}(\Theta)\leftarrow \text{PA}(\Theta)\vee (\text{CH}_1(\Theta)\vee \text{CH}_2(\Theta))$,
		\EndIf
		\EndFor
	\end{algorithmic}
\end{algorithm}

In Algorithm \textit{Backtracking}, $\text{PA}(\Theta)$ and $\text{CH}_1(\Theta), \text{CH}_2(\Theta)$ represent the parent and two children nodes of the Boolean operator node $\Theta\in \{\wedge, \vee\}$, respectively.

\begin{remark}\label{Rem1}
In \cite{gao2021}, the TLT is introduced for the model checking and control synthesis of discrete-time systems under LTL tasks. In this work, the sTLT is designed to guide the design of CBFs for continuous-time dynamical systems under nested STL formulas. The {much more complex} time constraints encoded in STL formulas have naturally led to different construction procedures and semantics of sTLT when compared to TLT in \cite{gao2021}, which we highlight as follows. First, the construction of sTLT is largely different from TLT. To incorporate the time constraints encoded in an STL formula, our construction of the sTLT relies on the finite time reachability analysis, i.e., the maximal and minimal reachablility operators $\mathcal{R}^M$ and $\mathcal{R}^m$ given respectively in Definitions \ref{Def:maxreachset} and \ref{Def:minreachset}. In \cite{gao2021}, the TLT construction relies on the infinite time controlled invariant set and robust controlled invariant set. Second, the sTLT semantics is largely different from TLT semantics. In order to monitor the
time constraint satisfaction in an STL formulas, we introduce in this work the definition of \emph{time interval coding} (cf. Definition \ref{Def:timecoding}) for a complete path of an sTLT. On one hand, we show in Definitions \ref{Def:PathSaf} and \ref{Def:TreeSaf} that the satisfaction of an sTLT can be characterized by the existence of a well-defined time interval coding. On the other hand, it will become clear later that the time interval coding is also vital in control synthesis. In \cite{gao2021}, the TLT semantics is much simpler as it only requires an assignment of ascending integers as time indices for each complete path of TLT.
\end{remark}

To better understand the sTLT semantics, i.e., Definition \ref{Def:TreeSaf}, the following definitions are needed.

\begin{definition} \label{def:top level}
	 We say \emph{an sTLT $\mathcal{T}$ contains $\vee$ operator nodes only at its top layers} if for every complete path $\bm{p}=\mathbb{X}_0\Theta_1\mathbb{X}_1\Theta_2\ldots \Theta_{N_f} \mathbb{X}_{N_f}$ of $\mathcal{T}$  that contains $\vee$ operator nodes, there exists a $1\le k\le N_f$ such that 
	 \begin{equation}\label{theta}
	     \Theta_{j} \in \begin{cases}
	 \{\vee\} & j\in \{1,\ldots, k\}, \\
		 \{ \wedge, F_{[a,b]}, G_{[a,b]}\}, & j\in \{k+1,\ldots, N_f\}.
		\end{cases}
	 \end{equation}
\end{definition}

\begin{remark}\label{rem:toplayer}
 For any nested STL formula $\varphi$, the operator nodes $\vee$, if it exists, only appears in the top layers of the constructed sTLT $\mathcal{T}_{\hat{\varphi}}$. This can be seen from the fact that $\hat\varphi$ is in the form of $\hat{\varphi} = \varphi_1 \vee \varphi_2 \vee \ldots \vee \varphi_N$, and $\varphi_i, i= 1, 2, \ldots, N,$ contain no $\vee$ operator as discussed in Step 1.
\end{remark}

\begin{definition} \label{def:branch}
Let $\bm{p}_l =\mathbb{X}_0 \Theta_1^{l} \mathbb{X}_1^{l} \Theta_2^{l}\ldots \Theta_{N_f}^{l} \mathbb{X}_{N_f}^{l}$ and $\bm{p}_f = \mathbb{X}_0\Theta_1^f \mathbb{X}_1^f \Theta_2^f \ldots \Theta_{N'_f}^f \mathbb{X}_{N'_f}^f$ be two complete paths of an sTLT $\mathcal{T}$. Denote by $k_l=\arg\max_{k}\{\Theta_k^l=\vee\}$ and $k_f=\arg\max_{k}\{\Theta_k^f=\vee\}$. We say \emph{$\bm{p}_l$ and $\bm{p}_f$ belong to the same branch of $\mathcal{T}$} if $k_l=k_f$ and $\mathbb{X}_j^{l} = \mathbb{X}_j^{f}, \Theta_j^{l} = \Theta_j^{f}, \forall j=1, \ldots k_l$. 
\end{definition}

\begin{remark} \label{rem:treeSatisfy}
   Definition \ref{Def:TreeSaf} can be interpreted as follows: 
   \begin{itemize}
       \item[1)] Consider the case where the sTLT $\mathcal{T}$ contains no $\vee$ operator. Then Definition \ref{Def:TreeSaf} dictates that $(\bm x,t) \cong  \mathcal{T}$ if and only if $(\bm x,t)$ satisfies every complete path of $\mathcal{T}$. 
   \item[2)] 
        Consider the case where the sTLT $\mathcal{T}$ contains $\vee$ operator nodes only at its top layers.  Then Definition \ref{Def:TreeSaf} dictates that $(\bm x,t) \cong  \mathcal{T}$ if and only if $(\bm x,t)$ satisfies at least one branch of complete paths.
   
   \end{itemize}

\end{remark}

\begin{example*}[continued] \label{ex:complete paths}
Let us continue with Example \ref{example1}. According to Definition \ref{def:completepath},
the sTLT $\mathcal{T}_{\hat\varphi}$ (see Fig. \ref{Fig:example}) has in total 3 complete paths, i.e., 
\begin{align*}
    \bm{p}_1&=\mathbb{X}_0\vee \mathbb{X}_1 \mathsf{F}_{[0, 15]}\mathbb{X}_3\mathsf{G}_{[2, 10]}\mathbb{X}_5, \\
    \bm{p}_2&=\mathbb{X}_0\vee \mathbb{X}_2 \mathsf{F}_{[0, 15]}\mathbb{X}_4\wedge\mathbb{X}_6 \mathsf{G}_{[0, 10]}\mathbb{X}_8, \\
    \bm{p}_3&=\mathbb{X}_0\vee \mathbb{X}_2 \mathsf{F}_{[0, 15]}\mathbb{X}_4\wedge\mathbb{X}_7\mathsf{F}_{[5, 10]}\mathbb{X}_9,
\end{align*}
 and 5 temporal fragments, which are encircled by the red dashed rectangles in Fig. \ref{Fig:example}. 
 
 The sTLT $\mathcal{T}_{\hat\varphi}$ contains $\vee$ operator nodes only at its top layers since one has $k_1=k_2=k_3=1$ according to Definition \ref{def:top level}. On one hand, one observes that $\Theta_1^2=\Theta_1^3=\vee$ and $\mathbb{X}_1^2=\mathbb{X}_1^3=\mathbb{X}_2$. Therefore, $\bm{p}_2$ and $\bm{p}_3$ belong to the same branch. On the other hand, since $\mathbb{X}_1^1=\mathbb{X}_1 \neq  \mathbb{X}_2 = \mathbb{X}_1^2 =\mathbb{X}_1^3 $, neither $\bm{p}_1$ and $\bm{p}_2$ nor $\bm{p}_1$ and $\bm{p}_3$  belong to the same branch. A trajectory $(\bm{x}, t)\cong\mathcal{T}_{\hat\varphi}$ if and only if either of the following 2 conditions is satisfied: (1) $(\bm{x}, t)\cong\bm{p}_1$, (2) $(\bm{x}, t)\cong\bm{p}_2$ and $(\bm{x}, t)\cong\bm{p}_3$. 
\end{example*}

\subsection{Relations between $\mathcal{T}_{\hat{\varphi}}$ and $\hat\varphi$ ($\varphi$)}

In this section, we derive the relations between an STL formula $\hat\varphi$ ($\varphi$) and its constructed sTLT $\mathcal{T}_{\hat{\varphi}}$.
First, we show the result for STL formulas in desired form, i.e., $\hat{\varphi}$.

\begin{theorem}\label{thm1}
Consider the system (\ref{x0}) and an STL task $\hat\varphi$ in desired form as per Definition \ref{def:desired}. The sTLT $\mathcal{T}_{\hat\varphi}$ is equivalent to $\hat\varphi$ in the sense that
\begin{equation}\label{desired_equa}
    (\bm{x}, t)\cong \mathcal{T}_{\hat\varphi}\Leftrightarrow (\bm{x}, t)\models \hat\varphi.
\end{equation}
\end{theorem}

\begin{proof}
For $\top$, predicate $\mu$, its negation $\neg \mu$, $\mu_1\wedge \mu_2$, and $\mu_1\vee \mu_2$, it is trivial to verify that $(\bm{x}, t)\cong \mathcal{T}_{\hat\varphi}\Leftrightarrow (\bm{x}, t)\models \hat\varphi.$

Next, we follow the induction rule to show that if $(\bm{x}, t)\cong \mathcal{T}_{\hat\varphi_1}\Leftrightarrow (\bm{x}, t)\models \hat\varphi_1$ and $(\bm{x}, t)\cong \mathcal{T}_{\hat\varphi_2}\Leftrightarrow (\bm{x}, t)\models \hat\varphi_2$, then the constructed sTLT $\mathcal{T}_{\hat\varphi}$ satisfies $(\bm{x}, t)\cong \mathcal{T}_{\hat\varphi}\Leftrightarrow (\bm{x}, t)\models \hat\varphi$ for a) $\hat\varphi=\hat\varphi_1 \wedge \hat\varphi_2, $ b) $  \hat\varphi_1 \vee \hat\varphi_2, $ c)$\mathsf{F}_{[a, b]} \hat\varphi_1, $ and d) $\mathsf{G}_{[a, b]} \hat\varphi_1$.

Case a): $\hat\varphi=\hat\varphi_1 \wedge \hat\varphi_2$. Assume that a trajectory $(\bm{x}, t)\cong \mathcal{T}_{\hat\varphi}$.  According to Definition \ref{Def:TreeSaf}, we have  $(\bm{x}, t)\cong \mathcal{T}_{\hat\varphi_1}$ and $(\bm{x}, t)\cong \mathcal{T}_{\hat\varphi_2}$. Under the assumption that $(\bm{x}, t)\cong \mathcal{T}_{\hat\varphi_1}\Leftrightarrow (\bm{x}, t)\models \hat\varphi_1$ and $(\bm{x}, t)\cong \mathcal{T}_{\hat\varphi_2}\Leftrightarrow (\bm{x}, t)\models \hat\varphi_2$, one can get that $(\bm{x}, t)\models \hat\varphi_1$ and $(\bm{x}, t)\models \hat\varphi_2$, which implies $(\bm{x}, t)\models \hat\varphi$. Thus, $(\bm{x}, t)\cong \mathcal{T}_{\hat\varphi}\Rightarrow (\bm{x}, t)\models \hat\varphi$. The proof of the other direction
is similar and hence omitted. 

Case b): $\hat\varphi=\hat\varphi_1 \vee \hat\varphi_2$. The proof is similar to Case a) and hence omitted.

Case c): $\hat\varphi=\mathsf{F}_{[a, b]}\hat\varphi_1$. Assume that a trajectory $(\bm{x}, t)\cong\mathcal{T}_{\mathsf{F}_{[a, b]}\hat\varphi_1}$. As depicted in Fig.1(c), we know that each complete path of $\mathcal{T}_{\hat\varphi}$ can be written in the form of $\bm{p} = \mathbb{X}_0 \Theta_1 \bm{p}^{\prime}$, where $ \Theta_1=F_{[a,b]}$ and $\bm{p}^{\prime}$ is a complete path of $\mathcal{T}_{\hat{\varphi}_1}$. According to Definitions \ref{Def:PathSaf} and \ref{Def:TreeSaf}, we have $\exists t'\in [t+a, t+b], \bm{x}(t')\in \mathbb{S}_{\hat\varphi_1}$ and $(\bm{x}, t')\cong \mathcal{T}_{\hat\varphi_1}$.  Under the assumption that $(\bm{x}, t)\cong \mathcal{T}_{\hat\varphi_1}\Leftrightarrow (\bm{x}, t)\models \hat\varphi_1$, one can get that $\exists t'\in [a, b], (\bm{x}, t')\models \hat\varphi_1$, which implies $(\bm{x}, t)\models \mathsf{F}_{[a, b]}\hat\varphi_1$ by Definition \ref{STLsemantics}. Thus, $(\bm{x}, t)\cong\mathcal{T}_{\mathsf{F}_{[a, b]}\hat\varphi_1}\Rightarrow (\bm{x}, t)\models \mathsf{F}_{[a, b]}\hat\varphi_1$. Assume now that $(\bm{x}, t)\models \mathsf{F}_{[a, b]}\hat\varphi_1$. Then one has from Definition \ref{STLsemantics} that $\exists t'\in [t+a, t+b], \bm{x}(t')\in \mathbb{S}_{\hat\varphi_1}$, which implies $\bm{x}(t)\in \mathcal{R}^{M}(\mathbb{S}_{\hat\varphi_1}, [a, b])=\mathbb{X}_0$. According to Definitions \ref{Def:PathSaf} and \ref{Def:TreeSaf}, it means that $(\bm{x}, t)\cong \mathcal{T}_{\hat\varphi}$. Therefore, $(\bm{x}, t)\models \mathsf{F}_{[a, b]}\hat\varphi_1\Rightarrow (\bm{x}, t)\cong\mathcal{T}_{\mathsf{F}_{[a, b]}\hat\varphi_1}.$

Case d): $\hat\varphi=\mathsf{G}_{[a, b]}\hat\varphi_1$. Assume that a trajectory $(\bm{x}, t)\cong\mathcal{T}_{\mathsf{G}_{[a, b]}\hat\varphi_1}$. As depicted in Fig.1(d), we know that each complete path of $\mathcal{T}_{\hat\varphi}$ can be written in the form of $\bm{p} = \mathbb{X}_0 \Theta_1 \bm{p}^{\prime}$, where $ \Theta_1=\mathsf{G}_{[a,b]}$ and $\bm{p}^{\prime}$ is a complete path of $\mathcal{T}_{\hat{\varphi}_1}$.  According to Definitions \ref{Def:PathSaf}  and \ref{Def:TreeSaf}, we have 
$(\bm{x}, t')\cong \mathcal{T}_{\hat\varphi_1}, \forall t'\in [t+a, t+b]$. Under the assumption that $(\bm{x}, t)\cong \mathcal{T}_{\hat\varphi_1}\Leftrightarrow (\bm{x}, t)\models \hat\varphi_1$, one can get that $(\bm{x}, t')\models \hat\varphi_1, \forall t'\in [t+a, t+b]$, which implies $(\bm{x}, t)\models \mathsf{G}_{[a, b]}\hat\varphi_1$ by Definition \ref{STLsemantics}. Thus, $(\bm{x}, t)\cong\mathcal{T}_{\mathsf{G}_{[a, b]}\hat\varphi_1}\Rightarrow (\bm{x}, t)\models \mathsf{G}_{[a, b]}\hat\varphi_1$. Assume now that $(\bm{x}, t)\models \mathsf{G}_{[a, b]}\hat\varphi_1$. Then one has that $\bm{x}(t')\in \mathbb{S}_{\hat\varphi_1}, \forall t'\in [t+a, t+b]$. Since $\hat\varphi_1$ contains no ``disjunction" operator according to Definition 
\ref{def:desired}, one can further get that  $\bm{x}(t)\in \overline{\mathcal{R}^m(\overline{\mathbb{S}_{\hat\varphi_1}}, [a, b])}=\mathbb{X}_0$. According to Definitions \ref{Def:PathSaf} and \ref{Def:TreeSaf}, it means that $(\bm{x}, t)\cong \mathcal{T}_{\hat\varphi}$. Therefore, $(\bm{x}, t)\models \mathsf{G}_{[a, b]}\hat\varphi_1\Rightarrow (\bm{x}, t)\cong\mathcal{T}_{\mathsf{G}_{[a, b]}\hat\varphi_1}.$
\end{proof}

For general STL tasks, we have the following result.

\begin{theorem}\label{thm2}
	Consider the system (\ref{x0}) and an STL task $\varphi$ in Definition \ref{Def:PNF}. The sTLT $\mathcal{T}_{\hat\varphi}$ is an under-approximation of $\varphi$ in the sense that
	$$(\bm{x}, t)\cong \mathcal{T}_{\hat\varphi}\Rightarrow (\bm{x}, t)\models \varphi.$$
\end{theorem}

\begin{proof}
The proof can be completed by showing 1) $(\bm{x}, t)\cong \mathcal{T}_{\hat\varphi}\Leftrightarrow (\bm{x}, t)\models \hat\varphi$ and 2) $(\bm{x}, t)\models \hat\varphi\Rightarrow (\bm{x}, t)\models \varphi.$ The proof for condition 1) is given in Theorem 1. Condition 2) is straightforward since $(\bm{x}, t)\models \mathsf{G}_{[0, b]}\varphi_1 \wedge \mathsf{F}_{[a, b]}\varphi_2 \Rightarrow (\bm{x}, t)\models \varphi_1\mathsf{U}_{[a, b]}\varphi_2$, $(\bm{x}, t)\models \mathsf{G}_{[a, b]}\varphi_1\vee \mathsf{G}_{[a, b]}\varphi_2 \Rightarrow (\bm{x}, t)\models \mathsf{G}_{[a, b]}(\varphi_1\vee\varphi_2)$, and $(\bm{x}, t)\models \mathsf{F}_{[a, b]}\varphi_1\vee \mathsf{F}_{[a, b]}\varphi_2 \Leftrightarrow (\bm{x}, t)\models \mathsf{F}_{[a, b]}(\varphi_1\vee\varphi_2)$, one has from the construction of the sTLT, Step 1 that $(\bm{x}, t)\models \hat\varphi\Rightarrow (\bm{x}, t)\models\varphi$.
The conclusion follows.
\end{proof}

\begin{remark}
   It becomes apparent from Theorems \ref{thm1} and \ref{thm2} that the under-approximation gap between general STL formula $\varphi$ in Definition \ref{Def:PNF} and sTLT is a result of \emph{Step 1} of constructing the sTLT, i.e., rewrite the STL formula $\varphi$ into the desired form $\hat\varphi$. In this step, we conduct two operations. First, we rewrite ``until" operator $\varphi=\varphi_1\mathsf{U}_{[a, b]} \varphi_2$ as $\hat\varphi=\mathsf{G}_{[0, b]}\varphi_1 \wedge \mathsf{F}_{[a, b]}\varphi_2$. This operation introduces conservatism because $\hat\varphi$ requires the pre-argument $\varphi_1$ of the ``until" operator to be satisfied for the time interval $[0, b]$ while $\varphi$ requires that $\exists t'\in [a, b]$ such that $\varphi_1$ is satisfied for the  time interval $[0, t']$. Second, we rewrite $\varphi = \Theta_{[a, b]}(\varphi_1 \vee \varphi_2)$ as $\hat\varphi = \Theta_{[a, b]}\varphi_1 \vee \Theta_{[a, b]}\varphi_2$, where $\theta \in \{\mathsf{F}, \mathsf{G}\}$. In this operation, the conservatism comes from ``always" operator because $(\bm{x}, t)\models \mathsf{G}_{[a, b]}\varphi_1\vee \mathsf{G}_{[a, b]}\varphi_2 \Rightarrow (\bm{x}, t)\models \mathsf{G}_{[a, b]}(\varphi_1\vee\varphi_2)$ while the other direction does not hold in general. For ``eventually" operator, though, this operation introduces no conservatism. {We further note that rewriting an ``until" operator as a conjunction of an ``always" operator and an ``eventually" operator is also used in the CBF-based approaches for non-nested STL formulas  \cite{lindemann2018,lindemann2019}, whereas the ``disjunction" operator is not addressed. Therefore, our approach is no more conservative compared to existing CBF-based approaches.}
   \end{remark}

\subsection{Design of CBFs}

In this subsection, we will use the sTLT $\mathcal{T}_{\hat{\varphi}}$ to guide the CBF design for a given STL formula $\varphi$. The intuition is that, we will design a time interval encoding and appropriate CBFs to enforce the synthesized system trajectory to satisfy the sTLT $\mathcal{T}_{\hat{\varphi}}$ as per the conditions given in Definitions \ref{Def:PathSaf} and \ref{Def:TreeSaf}.

\noindent \subsubsection{Time encoding for the sTLT}

Before proceeding, the following notations are needed. Given an STL operator $\Theta \in \{\wedge, \vee,  \mathsf{F}_{[a, b]}, \mathsf{G}_{[a, b]}\}$, define the possible start time (interval) of $\Theta$ (i.e., time to evaluate the satisfaction of $\varphi_1 \Theta \varphi_2$ or $\Theta \varphi$) as
	\begin{equation}\label{timecoding1}
		[\underline{t}(\Theta), \bar{t}(\Theta)]:=\begin{cases}
			[0, 0], & \mbox{if } \Theta\in \{\wedge, \vee\}, \\
			[a, b], & \mbox{if } \Theta\in \{\mathsf{F}_{[a, b]}\},\\
			[a, a], & \mbox{if } \Theta\in \{ \mathsf{G}_{[a, b]}\}.
		\end{cases}
	\end{equation}
The start time for logic operators $\wedge$ and $\vee$ is 0. For the temporal operator $\mathsf{G}_{[a,b]}$, the start time is $a$. Note that for the temporal operator $\mathsf{F}_{[a,b]}$, any time instant in the interval $[a, b]$ fulfills item ii) of Definition \ref{Def:PathSaf}. To accommodate this uncertainty, we set the start time for $\mathsf{F}_{[a,b]}$ to be the interval $[a, b]$. 
 
In addition, we define the duration of $\Theta$ as
	\begin{equation}\label{timecoding1}
		\mathcal{D}(\Theta): =\begin{cases}
			0, & \mbox{if } \Theta\in \{\wedge, \vee, \mathsf{F}_{[a, b]}\}, \\
			b-a, & \mbox{if } \Theta\in \{ \mathsf{G}_{[a, b]}\}.
		\end{cases}
	\end{equation}

The root node of $\mathcal{T}_{\hat\varphi}$ is denoted by $\mathbb{X}_{\text{root}}$. Let $\mathbb{X}$ be the set which collects all the set nodes of the sTLT $\mathcal{T}_{\hat\varphi}$. For a set node $\mathbb{X}_i\in \mathbb{X}$, define $[\underline{t}_s(\mathbb{X}_i), \bar{t}_s(\mathbb{X}_i)]$ and $\mathcal{D}(\mathbb{X}_i)$ as the possible start time (interval) and the duration of $\mathbb{X}_i$, respectively. $\text{PA}(\mathbb{X}_i)$ denotes the parent of node $\mathbb{X}_i$. Therefore, one has that $\text{PA}(\mathbb{X}_i)$ is an operator node and $\text{PA}(\text{PA}(\mathbb{X}_i))$ is a set node.

Now, the calculation of the start time (interval) for each set node $\mathbb{X}_i$ (which is needed for ensuring the satisfaction of the sTLT $\mathcal{T}_{\hat\varphi}$ as shown in Theorem \ref{thm1}) is outlined in Algorithm \ref{alg:timeInterval}. 

\begin{algorithm}
\caption{\textit{calculateStartTimeInterval}} \label{alg:timeInterval}
\begin{algorithmic}[1]
\Require The sTLT $\mathcal{T}_{\hat\varphi}$.
\Ensure $\underline{t}_s(\mathbb{X}_i), \bar{t}_s(\mathbb{X}_i), \mathcal{D}(\mathbb{X}_i), \forall \mathbb{X}_i$.
\State $\begin{aligned}&\underline{t}_s(\mathbb{X}_{\text{root}})\leftarrow 
 0, \bar{t}_s(\mathbb{X}_{\text{root}})\leftarrow 
 0, \mathcal{D}(\mathbb{X}_{\text{root}})\leftarrow 
 0
 \end{aligned}$
\For {each non-root node $\mathbb{X}_i$ of $\mathcal{T}_{\varphi}$ through a top-down traversal,}
\State { $\begin{aligned}
&\underline{t}_{s}(\mathbb{X}_i)\leftarrow \underline{t}_{s}(\text{PA}(\text{PA}(\mathbb{X}_i)))+\underline{t}(\text{PA}(\mathbb{X}_i)), \\
& \bar{t}_{s}(\mathbb{X}_i)\leftarrow \bar{t}_{s}(\text{PA}(\text{PA}(\mathbb{X}_i)))+\bar{t}(\text{PA}(\mathbb{X}_i)),
\end{aligned}$}
\State {$\mathcal{D}(\mathbb{X}_i)\leftarrow \mathcal{D}(\text{PA}(\text{PA}(\mathbb{X}_i)))+\mathcal{D}(\text{PA}(\mathbb{X}_i)),$}
\EndFor
\end{algorithmic}
\end{algorithm}

Due to the uncertainty of the start time for temporal operator $\mathsf{F}_{[a,b]}$, one can see that the start times of some set nodes $\mathbb{X}_i$ may be unknown and belong to an interval after running Algorithm \ref{alg:timeInterval}. In the following, we show how to update the start times of such set nodes $\mathbb{X}_i$ online. 

We develop an event-triggered scheme to update the start times. For each set node $\mathbb{X}_i$ such that $\underline{t}_{s}(\mathbb{X}_i)\neq \bar{t}_{s}(\mathbb{X}_i)$, an event is triggered at time $t$ if:
\begin{equation}\label{trigger_condition}
    t \in [\underline{t}_{s}(\mathbb{X}_i)), \bar{t}_{s}(\mathbb{X}_i))] \;\wedge \; \bm{x}(t)\in \mathbb{X}_i.
\end{equation}
Once an event is triggered, we run Algorithm \ref{alg:onlineupdate} to update the start times of the set nodes. Note that once an event is triggered for a set node, its start time is fixed. 

\begin{algorithm}
\caption{\textit{onlineUpdate}} \label{alg:onlineupdate}
\begin{algorithmic}[1]
\Require The sTLT $\mathcal{T}_{\hat\varphi}$ and the triggering time $t$.
\Ensure the updated $\underline{t}_{s}(\mathbb{X}_i), \bar {t}_{s}(\mathbb{X}_i), \forall \mathbb{X}_i$.
\For {each $\mathbb{X}_i$ such that the triggering condition (\ref{trigger_condition}) is satisfied,}
\State { $\underline{t}_{s}(\mathbb{X}_i)\leftarrow t, \bar{t}_{s}(\mathbb{X}_i)\leftarrow t,$}
\EndFor
\For {each set node $\mathbb{X}_i$ such that $\underline{t}_{s}(\mathbb{X}_i)\neq \bar{t}_{s}(\mathbb{X}_i)$ through a top-down traversal,}
\State  {run line 3 of Algorithm \ref{alg:timeInterval},}
\EndFor
\end{algorithmic}
\end{algorithm}

\begin{example*}[continued]
    Let us continue with Example \ref{example1} to demonstrate the event-triggered online update scheme. 

    First, one can calculate the start time intervals for each set node $\mathbb{X}_i, i=\{0, 1, \cdots, 9\}$ in the sTLT $\mathcal{T}_{\hat\varphi}$ (see Fig. \ref{Fig:example}) according to Algorithm \ref{alg:timeInterval}, which give $[\underline{t}_s(\mathbb{X}_0), \bar{t}_s(\mathbb{X}_0)]=[\underline{t}_s(\mathbb{X}_1), \bar{t}_s(\mathbb{X}_1)]=[\underline{t}_s(\mathbb{X}_2), \bar{t}_s(\mathbb{X}_2)]=[0, 0]$, $[\underline{t}_s(\mathbb{X}_3), \bar{t}_s(\mathbb{X}_3)]=[\underline{t}_s(\mathbb{X}_4), \bar{t}_s(\mathbb{X}_4)]=[0, 15]$, $[\underline{t}_s(\mathbb{X}_5), \bar{t}_s(\mathbb{X}_5)]=[2, 17]$, $[\underline{t}_s(\mathbb{X}_6), \bar{t}_s(\mathbb{X}_6)]=[\underline{t}_s(\mathbb{X}_7), \bar{t}_s(\mathbb{X}_7)]=[0, 15], [\underline{t}_s(\mathbb{X}_8), \bar{t}_s(\mathbb{X}_8)]=[0, 15]$, and $[\underline{t}_s(\mathbb{X}_9), \bar{t}_s(\mathbb{X}_9)]=[5, 25]$. Note that due to the `eventually' operator $\mathsf{F}_{[0, 15]}$ which appears at the outermost layer of the nested STL formula $\varphi=\mathsf{F}_{[0, 15]}(\mathsf{G}_{[2, 10]}\mu_1 \vee \mu_2  \mathsf{U}_{[5, 10]}\mu_3)$, the start times of all the set nodes that belong to temporal fragments (i.e., $\mathbb{X}_i, i\in \{3, 4, 5, 8, 9\}$) are uncertain (i.e., belong to an interval). To reduce conservatism, we update the start time intervals of these set nodes online using the event-triggered scheme (\ref{trigger_condition}).

    Assume that at time instant $t=5s$, the event-triggered condition (\ref{trigger_condition}) is satisfied for set node $\mathbb{X}_4$, i.e., $5\in [\underline{t}_s(\mathbb{X}_4), \bar{t}_s(\mathbb{X}_4)]=[0, 15]$ and $\bm{x}(5)\in \mathbb{X}_4$, then Algorithm \ref{alg:onlineupdate} is activated. Following lines 1-3 of Algorithm \ref{alg:onlineupdate}, one has that $\underline{t}_s(\mathbb{X}_4)=\bar{t}_s(\mathbb{X}_4)=5$ (i.e., the start time of set node $\mathbb{X}_4$ is fixed). Then one can further fix the start times of the set nodes $\mathbb{X}_6=\mathbb{X}_7=\mathbb{X}_8$ (which are $5s$) and update the start time interval of the set node $\mathbb{X}_9$ as $[\underline{t}_s(\mathbb{X}_9), \bar{t}_s(\mathbb{X}_9)]=[10, 15]$.
\end{example*}

\subsubsection{CBF design for each temporal fragment} First, we have the following definition.

\begin{definition}
    We call a temporal fragment $f_j$ the \emph{predecessor} of another temporal fragment $f_i$ (or $f_i$ the \emph{successor} of $f_j$) if there exists a complete path $\bm p$ such that $\bm p = ... f_j \bm{p}^\prime f_i ...$ where  $\bm{p}^\prime$ does not contain any temporal fragments. We call $f_i$ a \emph{top-layer temporal fragment} if  $f_i$ has no predecessor temporal fragment.
\end{definition}

Given the sTLT $\mathcal{T}_{\hat\varphi}$ for a nested STL formula $\varphi$, we need to design one CBF for each temporal fragment $f_i$ in view of the item iv) of Definition \ref{Def:PathSaf}. Denote by $f_i=\Theta_{f_i}\mathbb{X}_{f_i}$, where $\Theta_{f_i}$ and $\mathbb{X}_{f_i}$ are the temporal operator node and the set node contained in $f_i$. Note that $\mathbb{X}_{f_i}$ is represented by its value function $\mathbb{X}_{f_i} = \{x: h_{\mathbb{X}_{f_i}}(x)\ge 0 \}$. We require the corresponding CBF $\mathfrak{b}_i(x, t)$ to satisfy the following conditions:
\begin{itemize}
  \item[1)] $\mathfrak{b}_i(x, t)$ is continuously differentiable and is defined over $\mathcal{C}(t)\times[\min\{{t}_e(\text{PA}(\text{PA}(\mathbb{X}_{f_i}))), \underline{t}_s(\mathbb{X}_{f_i})\}, t_e(\mathbb{X}_{f_i})]$;
  
  \item[2)] $\mathfrak{b}_i(x, t) \leq h_{\mathbb{X}_{f_i}}(x), \forall t\in [\bar{t}_s(\mathbb{X}_{f_i}), t_e(\mathbb{X}_{f_i})]$,
\end{itemize}
where $t_e(\mathbb{X}_{i})=\bar{t}_s(\mathbb{X}_{i})+\mathcal{D}(\mathbb{X}_i)$ {(recall $\mathcal{D}(\mathbb{X}_i)$ is computed in Algorithm \ref{alg:timeInterval})} can be interpreted as the end time of $\mathbb{X}_{i}$. Here $\underline{t}_s(\mathbb{X}_{f_i}),\bar{t}_s(\mathbb{X}_{f_i})$ and ${t}_e(\mathbb{X}_{f_i})$ are updated online according to Algorithm \ref{alg:onlineupdate}.

Define the \emph{time domain} of the CBF $\mathfrak{b}_i(x,t)$ as
\begin{equation}\label{eq:timedomain}  [\underline{t}_{\mathfrak{b}_i},\bar{t}_{\mathfrak{b}_i}]:=[\min\{{t}_e(\text{PA}(\text{PA}(\mathbb{X}_{f_i}))), \underline{t}_s(\mathbb{X}_{f_i})\}, t_e(\mathbb{X}_{f_i})].
\end{equation}
This is to guarantee that the CBF $\mathfrak{b}_i$, which corresponds to the temporal fragment $f_i$, is activated at  ${t}_e(\text{PA}(\text{PA}(\mathbb{X}_{f_i})))$, for which the activation of the predecessor of $f_i$ ends, or at  $\underline{t}_s(\mathbb{X}_{f_i})$, for which $f_i$ becomes active at its earliest, whichever comes earlier. A formal statement on this is given in Lemma \ref{lem:time_sequence}.

\begin{lemma} \label{lem:time_sequence}
    Let $f_i$ be a non top-layer temporal fragment, and $f_j$ be the predecessor of $f_i$ in the constructed sTLT. Denote their respective CBFs $\mathfrak{b}_j(x, t), \mathfrak{b}_i(x, t)$. Then $ \underline{t}_{\mathfrak{b}_j} \leq  \underline{t}_{\mathfrak{b}_i} \leq \bar{t}_{\mathfrak{b}_j} \leq \bar{t}_{\mathfrak{b}_i}  $.
\end{lemma}
\begin{proof}
It can be deduced from the tree structure that the predecessor of a non top-layer temporal fragment is unique. Denote the set nodes in the fragments $f_j$ and  $f_i$ are $\mathbb{X}_{f_j}, \mathbb{X}_{f_i}$, respectively. The inequalities can be obtained as follows: 1) in view of \eqref{eq:timedomain} and Algorithm \ref{alg:timeInterval}, $\underline{t}_{\mathfrak{b}_j} \leq  {t}_{e}(\mathbb{X}_{f_j})$ and $\underline{t}_{\mathfrak{b}_j} \leq  \underline{t}_{s}(\mathbb{X}_{f_j}) \leq \underline{t}_{s}(\mathbb{X}_{f_i}) $, thus $\underline{t}_{\mathfrak{b}_j} \leq \underline{t}_{\mathfrak{b}_i} = \min({t}_{e}(\mathbb{X}_{f_j}), \underline{t}_{s}(\mathbb{X}_{f_i}) )$; 2) from \eqref{eq:timedomain}, $\underline{t}_{\mathfrak{b}_i} \leq t_{e}(\mathbb{X}_{f_j})=\bar{t}_{\mathfrak{b}_j} $; 3) from Algorithm \ref{alg:timeInterval} and the definition of $t_e(\cdot)$, 
$\bar{t}_{\mathfrak{b}_j} =t_{e}(\mathbb{X}_{f_j}) = \bar{t}_s(\mathbb{X}_j) + \mathcal{D}(\mathbb{X}_j)\leq   \bar{t}_s(\mathbb{X}_i) + \mathcal{D}(\mathbb{X}_i) =  t_{e}(\mathbb{X}_{f_i}) =\bar{t}_{\mathfrak{b}_i} $.

\end{proof}

If $f_i$ is not a top-layer temporal fragment, then the third condition on the corresponding CBF $\mathfrak{b}_i(x, t)$ is
\begin{itemize}
  \item[3)] $  \mathfrak{b}_i(x,\underline{t}_{\mathfrak{b}_i}) \ge 0, \forall x\in \{x: \mathfrak{b}_j(x,\underline{t}_{\mathfrak{b}_i}) \ge 0\}$, where $f_j$ is the unique predecessor of $f_i$.
\end{itemize}
Note that $\mathfrak{b}_j(x,\underline{t}_{\mathfrak{b}_i})$ is well-defined in view of Lemma \ref{lem:time_sequence}.

\begin{proposition} \label{prop:cbf_for_complete_path}
Given a complete path $\bm p$ and an initial condition $x_0$, let $f_0, f_1, ..., f_N$ be the sequence of temporal fragments contained in $\bm p$ and $\mathfrak{b}_0, \mathfrak{b}_1, \ldots, \mathfrak{b}_N$ the corresponding CBFs. Assume that each $\mathfrak{b}_i, i\in 0, \ldots, N$ satisfies the conditions 1)-3). Furthermore, if $\mathfrak{b}_0(x_0,0)\ge 0$ and each of the CBFs $\mathfrak{b}_i$ satisfies the condition \eqref{eq:cbf_condition} during the corresponding time domain, then the resulting trajectory satisfies this complete path $\bm p$.
\end{proposition}

\begin{proof}
 Without loss of generality, assume that $f_i$ is the predecessor of $f_{i+1}$, $i=0,1,...,N-1$. For the top-level temporal fragment $f_0$, since $\mathfrak{b}_0(x_0,0)\ge 0$ and the CBF condition \eqref{eq:cbf_condition} holds in $[0,\bar{t}_{\mathfrak{b}_0}]$, we have $ \mathfrak{b}_0(\bm{x}(t),t) \ge 0, \forall t\in [0,\bar{t}_{\mathfrak{{b}}_0}]$. Now assume $\mathfrak{b}_i(\bm{x}(t),t) \ge 0,  \forall t\in [\underline{t}_{\mathfrak{b}_i},\bar{t}_{\mathfrak{b}_i}]$. From condition 3), $\mathfrak{b}_{i+1}(\bm{x}(\underline{t}_{\mathfrak{b}_{i+1}}),\underline{t}_{\mathfrak{b}_{i+1}}) \ge  0$. In addition, the CBF condition \eqref{eq:cbf_condition} of $\mathfrak{b}_{i+1}$ is satisfied for  $\forall t\in [\underline{t}_{\mathfrak{b}_{i+1}},\bar{t}_{\mathfrak{{b}_{i+1}}}]$, and then   $ \mathfrak{b}_{i+1}(\bm{x}(t),t) \ge 0, \forall t\in [\underline{t}_{\mathfrak{b}_{i+1}},\bar{t}_{\mathfrak{{b}_{i+1}}}]$. Inductively, we obtain $\mathfrak{b}_{i}(\bm{x}(t),t) \ge 0, \forall t \in [\underline{t}_{\mathfrak{b}_{i}},\bar{t}_{\mathfrak{{b}_{i}}}]$ for $i = 0,1,2, ..., N$.

In addition, $ \mathfrak{b}_{i}(\bm{x}(t),t) \ge 0, \forall t \in [\underline{t}_{\mathfrak{b}_{i}},\bar{t}_{\mathfrak{{b}_{i}}}]$ implies that $\bm{x}(t)\in \mathbb{X}_i, \forall t\in [\bar{t}_s(\mathbb{X}_{f_i}), {t}_e(\mathbb{X}_{f_i})]$ from condition 2). One verifies that $[\bar{t}_s(\mathbb{X}_{f_i}), {t}_e(\mathbb{X}_{f_i})], \forall f_i$  is a valid time interval coding of the complete path from Definition \ref{Def:PathSaf} items i-iii). Thus, the resulting trajectory satisfies the complete path $\bm p$.\end{proof}

Up to now, we have shown the design of the CBFs and the online update of their time domains for each temporal fragment in the sTLT $\mathcal{T}_{\hat{\varphi}}$. In what follows, we will show how to incorporate them to conduct the online control synthesis.

\subsection{The overall algorithm}

In this subsection, we divide the nested STL formulas into 2 classes, i.e., nested STL formulas that contain no $\vee$ operator and nested STL formulas that contain $\vee$ operator. We differentiate these two cases because they have different sTLT satisfaction conditions as discussed in Remark \ref{rem:treeSatisfy}. 

\subsubsection{Nested STL formulas that contain no $\vee$ operator}
Let $\varphi$ be a nested STL formula that contains no $\vee$ operator. Then, the corresponding sTLT  $\mathcal{T}_{\hat\varphi}$ contains no operator nodes $\vee$. Let $\Pi$ be the set which collects all the temporal fragments $f_i$. Denote by $\mathfrak{b}_i$ the CBF designed for the temporal fragment $f_i$.  Note that when the start time interval is updated online (Algorithm \ref{alg:onlineupdate}), the time domain of the CBF $\mathfrak{b}_i$ will also be updated correspondingly. The continuous-time control synthesis problem (Problem \ref{problem}) can be solved by the following program:
\begin{equation} \label{eq:online_control_wo_disjunction}
    \begin{aligned}
      &\hspace{0.5cm}\min_{u\in U} \quad u^TQu \\
  &\text{s.t.} \quad \theta_i(t)\Big(\frac{\partial \mathfrak{b}_i(x, t)}{\partial x}f(x, u) +\frac{\partial \mathfrak{b}_i(x, t)}{\partial t} \\
  &\hspace{3.5cm} +\alpha_i(\mathfrak{b}_i(x, t))\Big)\ge 0, \forall f_i\in \Pi,
    \end{aligned}
\end{equation} 
where $\theta_i(t) = \begin{cases} 
1, & \text{if } t\in [\underline{t}_{\mathfrak{b}_i},\bar{t}_{\mathfrak{b}_i}]\\
0, & \text{otherwise}
\end{cases}$ is an indicator function assigned to each CBF $\mathfrak{b}_i$. Note that since $\underline{t}_{\mathfrak{b}_i}, \bar{t}_{\mathfrak{b}_i}$ are updated online, $\theta_i(t)$ is also updated online.

\subsubsection{Nested STL formulas that contain $\vee$ operator} Let $\varphi$ be a nested STL formula that contains $\vee$ operators.
Then, as discussed in Remark \ref{rem:toplayer}, the operator nodes $\vee$ only appear in the top layers of $\mathcal{T}_{\hat{\varphi}}$.

Recall from Remark \ref{rem:treeSatisfy} that to obtain $(\bm x,0) \cong \mathcal{T}_{\hat{\varphi}}$, $(\bm x, 0)$ needs to satisfy at least one branch of complete paths. Deciding which group of complete paths to satisfy can be done offline or online. In the following we show the case where the branch is chosen offline.

Without loss of generality, let $\Pi_l$ be the set which collects all the temporal fragments $f_i$ that belongs to the chosen branch. Then the online control synthesis is given by
 \begin{eqnarray} \label{eq:online_syn_lth_group}
  &&\hspace{0.5cm}u = \argmin_{u\in U} \quad u^TQu\\
  &&\hspace{-0.5cm}\text{s.t.} \quad \theta_i(t)\Big(\frac{\partial \mathfrak{b}_i(x, t)}{\partial x}f(x, u) +\frac{\partial \mathfrak{b}_i(x, t)}{\partial t} \nonumber\\
  &&\hspace{3.5cm} +\alpha_i(\mathfrak{b}_i(x, t))\Big)\ge 0, \forall f_i\in \Pi_l,\nonumber \nonumber
\end{eqnarray}
where $\mathfrak{b}_i$ is the designed CBF according to $f_i$, $\theta_i(t) = \begin{cases}
1, & \text{if } t\in [\underline{t}_{\mathfrak{b}_i},\bar{t}_{\mathfrak{b}_i}]\\
0, & \text{otherwise}
\end{cases}$ is an indicator function assigned to each CBF $\mathfrak{b}_i$. Similar to the previous case, $\theta_i(t)$ is updated online by Algorithm \ref{alg:onlineupdate}.

\begin{remark}[Choice of branch]
    The constructed sTLT provides a general guideline on how to choose the branch to satisfy. For example, denote by  $\bm{p}_l =\mathbb{X}_0 \Theta_1^{l} \mathbb{X}_1^{l} \Theta_2^{l}\ldots \Theta_{N_f}^{l} \mathbb{X}_{N_f}^{l}$ an arbitrary complete path in branch $l$. Let $k_l = \argmax_k\{\Theta_k^l = \vee\}$. Then the branch $l$ can be chosen only if the initial state $x_0 \in \mathbb{X}_{k_l}^l$. This condition is evident from the sTLT semantics. One numerical example is given in Case Studies where only one branch out of two can be chosen. Several other factors can be considered when selecting the branch. For example, one can use performance indexes like robustness metrics, optimal energy, shortest path or online re-plan in the presence of environmental uncertainties. This is however out of the scope of this work and will be pursued in the future.
\end{remark}

\begin{remark}[Online CBFs update] \label{rem:online_updates}
   Even though the time domains of the offline designed CBFs change as the start time intervals update online, this does not impose a need to re-compute the barriers from scratch. Instead, a simple translation in time will suffice. To illustrate this point, assume that we have computed two barriers $\mathfrak{b}_j(x, t),  t\in [\underline{t}_{\mathfrak{b}_j},\bar{t}_{\mathfrak{b}_j}]$ and $\mathfrak{b}_i(x, t),  t\in [\underline{t}_{\mathfrak{b}_i},\bar{t}_{\mathfrak{b}_i}]$ for two consecutive temporal fragments $f_j f_i = \Theta_{f_j}\mathbb{X}_{f_j}\Theta_{f_i}\mathbb{X}_{f_i}$. Denote $ \bar{t}_s(\mathbb{X}_{f_j})$ before the update by $t_1$. If, at time $t^\prime\in [\underline{t}_s(\mathbb{X}_{f_j}), \bar{t}_s(\mathbb{X}_{f_j}) ]$, $\underline{t}_s(\mathbb{X}_{f_j})\neq \bar{t}_s(\mathbb{X}_{f_j}) $ and $\bm{x}(t^\prime)$ reaches $\mathbb{X}_{f_j}$, then Algorithm \ref{alg:onlineupdate} updates $\underline{t}_s(\mathbb{X}_{f_j})= \bar{t}_s(\mathbb{X}_{f_j}) = t^\prime $, and, accordingly, the new time domains of the barriers become $[\underline{t}^\prime_{\mathfrak{b}_j},\bar{t}^\prime_{\mathfrak{b}_j}] :=[t^\prime,t^\prime+\mathcal{D}(\mathbb{X}_{f_j})]$ and $[\underline{t}^\prime_{\mathfrak{b}_i},\bar{t}^\prime_{\mathfrak{b}_i}]:=
   [\underline{t}_{\mathfrak{b}_i}+t^\prime - t_1,\bar{t}_{\mathfrak{b}_i} + t^\prime - t_1]$. The updated barriers are $\mathfrak{b}_j^\prime(x, t) = \mathfrak{b}_j(x, t+t_1 - t^\prime),  t\in [\underline{t}^\prime_{\mathfrak{b}_j},\bar{t}^\prime_{\mathfrak{b}_j}]$ and $\mathfrak{b}_i^\prime(x, t) =\mathfrak{b}_i(x, t+t_1 - t^\prime),  t\in [\underline{t}^\prime_{\mathfrak{b}_i},\bar{t}^\prime_{\mathfrak{b}_i}]$, respectively.

\end{remark}

\begin{remark} \label{rem:extension to full STL}
  Recall that the above analysis is done for nested STL formulas as per Definition \ref{nestedSTL}. It is straightforward to extend the results to STL tasks that are given by conjunction and/or disjunction of nested STL formulas, for instance, $\varphi=\mathsf{F}_{[0, 15]}(\mathsf{G}_{[0, 10]}\mu_1 \vee \mu_2  \mathsf{U}_{[5, 10]}\mu_3) \wedge \mathsf{G}_{[a_5, b_5]}\mu_4$. The sTLT thus is constructed for $\hat{\varphi} = \mathsf{F}_{[0, 15]}\mathsf{G}_{[2, 10]}\mu_1 \vee  \mathsf{F}_{[0, 15]}(\mathsf{G}_{[0, 10]}\mu_2  \wedge \mathsf{F}_{[5, 10]}\mu_3) \wedge \mathsf{G}_{[a_5, b_5]}\mu_4$.  The implementation can be done by adding an extra barrier condition corresponding to $\mathsf{G}_{[a_5, b_5]}\mu_4$ into  \eqref{eq:online_syn_lth_group}.
\end{remark}

Now we summarize our proposed solution in the following theorem.
\begin{theorem} \label{thm3}
    Consider a dynamical system \eqref{x0} and a nested STL specification $\varphi$. Let the sTLT constructed according to Section \ref{subsec:treeconstruction}. If the initial condition $x_0\in \mathbb{X}_{\text{root}}$ and the online program is feasible, then the resulting system trajectory $(\bm x,0)\models \varphi$.
\end{theorem}
\begin{proof}
The proof follows from Proposition \ref{prop:cbf_for_complete_path} and Theorem \ref{thm2}.
\end{proof}

\begin{remark}[Nominal control as heuristics]
    In literature dedicated to studying CBFs \cite{ames2016control}, it is common to incorporate nominal  controls to improve overall performance. This is usually done by replacing the weighted quadratic cost in  \eqref{eq:online_control_wo_disjunction} and \eqref{eq:online_syn_lth_group} to the form $(u-u_{nom})^T Q (u-u_{nom}) $, where $u_{nom}$ is usually designed based on heuristics. More details on designing $u_{nom}$ will be detailed in Case Studies.
 \end{remark}

\begin{remark}[Online feasibility] \label{rem:online_qp}
Although we require each $\mathfrak{b}_i$ to be a valid CBF, in general there is no guarantee that they are compatible \cite{xiao2022compatibility}, i.e., the online programs in  \eqref{eq:online_control_wo_disjunction} and  \eqref{eq:online_syn_lth_group} are feasible for all $(x,t)$. When the system is control-affine, the feasibility of QPs is guaranteed when the time domains of individual CBFs do not overlap. In the general case, one can verify or falsify the compatibility of multiple CBFs \textit{a priori} using the method from \cite{xiao2022compatibility}. More detailed discussions are given in Case Studies with several empirical remedies.
\end{remark}

\subsection{Computational complexity}

The computational complexity of the overall approach involves offline and online computational complexities. The offline phase is composed of 1) the construction of the sTLT and 2) the design of a CBF for each temporal fragment of the sTLT.

\textbf{Construction of sTLT:} Given an STL formula $\varphi$ in desired form with $K$ operators, the constructed sTLT contains at most $3K+1$ nodes ($K$ operator nodes and at most $2K+1$ set nodes). The bottleneck for constructing the sTLT, however, is the computation of set nodes, which involves computing maximal or minimal reachable sets (defined in Definitions \ref{Def:maxreachset} and \ref{Def:minreachset}) for the continuous-time dynamical systems under consideration. In the case of a linear continuous-time system, one can compute reachable sets efficiently for large-scale linear systems with several thousand state variables for bounded, but arbitrarily varying inputs \cite{althoff2019reachability}. In the case of a nonlinear continuous-time system, the computation of backward reachable sets is in general undecidable \cite{althoff2021set}. Fortunately, over the past decade, new approaches (e.g., decomposition approach \cite{chen2018decomposition} and deep learning approach \cite{bansal2021deepreach} and software tools (e.g.,
Hamilton-Jacobi Toolbox \cite{mitchell2005toolbox}
and CORA Toolbox \cite{althoff2015introduction}), have been developed for improving the efficiency of computing backward reachable sets. Once the sTLT is constructed, the design of a CBF further requires calculating the start time interval and duration of the set node in the corresponding temporal fragment (i.e., Algorithm 3). The complexity of Algorithm 3 is $\mathcal{O}(1)$. 

\textbf{Construction of CBFs:} The construction of CBFs can be computationally expensive for general nonlinear systems. Luckily, there are several remedies to simplify the computations. In view of the satisfaction condition of an sTLT, we could always {construct a CBF based on an under-approximation of the set node
in the corresponding temporal fragment when the exact reachable sets are difficult to calculate.} Moreover, if the system is fully actuated, the CBF can in general be constructed analytically. One such example is the single integrator dynamics shown in the Case Studies. Other approaches include sum-of-squares techniques \cite{wang2022safety}, learning-based approaches\cite{abate2021fossil}, and HJB reachability-based approaches \cite{wiltz2023construction}. In particular, we highlight that the construction of CBFs through HJB reachability analysis is a byproduct of computing the maximal/minimal reachable sets, which are essential for building the sTLT. The HJB reachability approach is also demonstrated in the Case Studies with unicycle dynamics.

\textbf{Online computations:} The online phase is composed of 1) the online update of the CBFs and 2) solving the optimization program (\ref{eq:online_control_wo_disjunction}) or (\ref{eq:online_syn_lth_group}). As pointed out in Remark \ref{rem:online_updates}, a simple translation in time is sufficient for updating the CBFs. Therefore, the complexity of this step is determined by online updating the start time intervals of set nodes (i.e., Algorithm 4), which is $\mathcal{O}(1)$. The complexity of the optimization program (\ref{eq:online_control_wo_disjunction}) or (\ref{eq:online_syn_lth_group}) is determined by the system model. When the continuous-time dynamical system (\ref{x0}) is control-affine, i.e., (\ref{x0}) is of the form (\ref{control-affine}), the programs (\ref{eq:online_control_wo_disjunction}) and (\ref{eq:online_syn_lth_group}) are QPs.

\section{Case studies}
In this section, we explain the explicit procedures to construct CBFs and formulate the online QP for the nested STL specification given in Example \ref{example1}. It is worth noting that the developed theory is dynamics agnostic. We will show this by designing control synthesis schemes for both single-integrator dynamics and unicycle dynamics, where analytical and numerical CBFs are constructed, respectively. In the end of this section, we demonstrate the efficacy of our proposed method under a more complex STL specification. All the implementation code can be found at \url{https://github.com/xiaotan-git/sTLT}. 

\subsection{Single integrator model}

Consider a mobile robot with a single-integrator dynamics 
\begin{equation} \label{eq:single integrator}
 \dot{x}=u,
\end{equation}
where $x=(x_1,x_2)\in \mathbb{R}^2$ and $u = (u_1, u_2)\in U\subset \mathbb{R}^2$, and the control input set $U=\{u: |u_1|\le 1, |u_2|\leq 1\}$. The STL task specification is given by $\varphi=\mathsf{F}_{[0, 15]}(\mathsf{G}_{[2, 10]}\mu_1 \vee \mu_2  \mathsf{U}_{[5, 10]}\mu_3)$ (the same as in Example \ref{example1}), where $\mathbb{S}_{\mu_1}=\{x\in \mathbb{R}^2\mid  (x_1+4)^2+(x_2+4)^2\leq 1\}$,  $\mathbb{S}_{\mu_2}=\{x\in \mathbb{R}^2\mid  (x_1-4)^2+x_2^2\leq 4^2\}$, and $\mathbb{S}_{\mu_3}=\{x\in \mathbb{R}^2\mid (x_1-1)^2+(x_2+4)^2\leq  2^2 \}$. Recall from Example \ref{example1}, the sTLT $\mathcal{T}_{\hat\varphi}$ is plotted in Fig. \ref{Fig:example}. 

One observation is that, for single integrator dynamics and a given set node $\mathbb{X}_{\varphi_1}$, the sets $\mathcal{R}^M(\mathbb{X}_{\varphi_1}, [a, b])$ and  $\overline{\mathcal{R}^m(\overline{\mathbb{X}_{\varphi_1}}, [a, b])}$, which are the set nodes obtained using the temporal operators $\mathsf{F}_{[a, b]}$ and $\mathsf{G}_{[a, b]}$ respectively, are monotonic increasing with  respect to the input set $U$. Thus, to simplify the set calculation, we calculate subsets of the reachable sets by shrinking the input set $U$ to $U^\prime = \{ u: \| u \|\leq 1\}$. Then one can get that
\begin{equation}
    \begin{aligned}
      & \mathbb{X}_5=\mathbb{S}_{\mu_1}, \mathbb{X}_8=\mathbb{S}_{\mu_2}, \mathbb{X}_9=\mathbb{S}_{\mu_3},\\
  & \mathbb{X}_3=\{x\in \mathbb{R}^2\mid  (x_1+4)^2+(x_2+4)^2\leq 3^2\},\\
  & \mathbb{X}_4=\{x\in \mathbb{R}^2\mid  (x_1-4)^2+x_2^2\leq 4^2\}, \\
  & \mathbb{X}_1=\{x\in \mathbb{R}^2\mid  (x_1+4)^2+(x_2+4)^2\leq 18^2\}, \\
  & \mathbb{X}_2=\{x\in \mathbb{R}^2\mid  (x_1-4)^2+x_2^2\leq 19^2\},\\
  & \mathbb{X}_0=\{x\in \mathbb{R}^2\mid  (x_1+4)^2+(x_2+4)^2\leq 18^2  \text{ or }  \\
  & \hspace{3cm} (x_1-4)^2+x_2^2\leq 19^2\}. \\
    \end{aligned}
\end{equation}
Here $\mathbb{X}_0, ...,\mathbb{X}_5$
are subsets of what one could obtain with the input set $U$. Yet the under-approximation relation still holds in view of the iv)th condition in Definition \ref{Def:PathSaf}. Here we note that although the sets $\mathbb{X}_0,\mathbb{X}_1,\mathbb{X}_2$ are not needed for CBF design (since they do not correspond to any temporal fragments), they still play an important role that will become clear later. The sets $\mathbb{X}_3, \mathbb{X}_4, \mathbb{X}_5,\mathbb{X}_8,\mathbb{X}_9$ are depicted in Fig. \ref{fig:traj single integrator}.

Denote the temporal fragments $f_1=\mathsf{F}_{[0, 15]} \mathbb{X}_3, f_2=\mathsf{G}_{[2, 10]} \mathbb{X}_{5}, f_3=\mathsf{F}_{[0, 15]} \mathbb{X}_4, f_4=\mathsf{G}_{[0, 10]} \mathbb{X}_{8}, f_5=\mathsf{F}_{[5, 10]} \mathbb{X}_{9} $ and their corresponding control barrier functions $\mathfrak{b}_1, ..., \mathfrak{b}_5$. Using Algorithm \ref{alg:timeInterval} and \eqref{eq:timedomain}, one obtainsthe initial starting time interval, the duration, and the time domain of the corresponding CBFs:
\begin{itemize} 
\itemindent=-12pt
     \item $[\underline{t}_s(\mathbb{X}_{3}), \bar{t}_s(\mathbb{X}_{3})]=[0, 15], \mathcal{D}(\mathbb{X}_{3}) = 0, [\underline{t}_{\mathfrak{b}_1},\bar{t}_{\mathfrak{{b}}_1}] = [0, 15]$;
	\item $[\underline{t}_s(\mathbb{X}_{5}), \bar{t}_s(\mathbb{X}_{5})]=[2, 17], \mathcal{D}(\mathbb{X}_{5}) = 8,[\underline{t}_{\mathfrak{b}_2},\bar{t}_{\mathfrak{{b}}_2}] = [2, 25]$;
	\item $[\underline{t}_s(\mathbb{X}_{4}), \bar{t}_s(\mathbb{X}_{4})]=[0, 15], \mathcal{D}(\mathbb{X}_{4}) = 0,[\underline{t}_{\mathfrak{b}_3},\bar{t}_{\mathfrak{{b}}_3}] = [0, 15]$;
	\item $[\underline{t}_s(\mathbb{X}_{8}), \bar{t}_s(\mathbb{X}_{8})]=[0, 15], \mathcal{D}(\mathbb{X}_{8}) = 10,[\underline{t}_{\mathfrak{b}_4},\bar{t}_{\mathfrak{{b}}_4}] = [0, 25]$;
	\item $[\underline{t}_s(\mathbb{X}_{9}), \bar{t}_s(\mathbb{X}_{9})]=[5, 25], \mathcal{D}(\mathbb{X}_{9}) = 0,[\underline{t}_{\mathfrak{b}_5},\bar{t}_{\mathfrak{{b}}_5}] = [5, 25]$.
\end{itemize}

Taking into account the velocity limit, we design the initial CBFs as
\begin{equation}  \label{eq:cbfs_single_integrator}
\begin{aligned}
    & \mathfrak{{b}}_1(x,t) =  (18-t)^2 - (x_1+4)^2-(x_2+4)^2, t\in [0,15]; \\
    & \mathfrak{{b}}_2(x,t) =  \begin{cases}
    (18-t)^2 - (x_1+4)^2-(x_2+4)^2, t\in [2,17]; \\
    1^2 - (x_1+4)^2-(x_2+4)^2, t\in [17,25];
    \end{cases} \\
     & \mathfrak{{b}}_3(x,t) =  
    (19-t)^2 - (x_1-4)^2-x_2^2, t\in [0,15]; \\
    & \mathfrak{{b}}_4(x,t) =  \begin{cases}
    (19-t)^2 - (x_1-4)^2-x_2^2, t\in [0,15]; \\
        4^2 - (x_1-4)^2-x_2^2, t\in [15,25];
    \end{cases} \\
    & \mathfrak{{b}}_5(x,t) =  (27-t)^2 - (x_1 - 1)^2 - (x_2+4)^2, t\in [5,25].
\end{aligned}
\end{equation}
It is evident that the zero super-level sets of the barriers are circular, which either remain static or shrink in radius at a velocity of $1$. If the robot is about to leave the safe region, i.e., when $\mathfrak{{b}}_i(x,t) = 0$, the robot can always steer itself towards the center with unit velocity, and thus always stay safe.
One could easily verify that, for $i= 1,2,...,5$, 1) $\mathfrak{b}_i(x, t)$ is a valid CBF for the single integrator dynamics in \eqref{eq:single integrator}; 2)  $\mathfrak{b}_i(x, t) = h_{\mathbb{X}_{f_i}}(x), \forall t\in [\bar{t}_s(\mathbb{X}_{f_i}), {t}_e(\mathbb{X}_{f_i})]$, where $\mathbb{X}_{f_i}$ is the set node in the corresponding temporal fragment $f_i$; 3) $  \mathfrak{b}_i(x,\underline{t}_{\mathfrak{b}_i}) \ge \mathfrak{b}_j(x,\underline{t}_{\mathfrak{b}_i}), \forall x$, where the corresponding temporal fragment $f_j$ is the predecessor of $f_i$. Thus, CBFs in \eqref{eq:cbfs_single_integrator} fulfill the conditions in Sec. III.D. Note that here we calculate the initial CBFs, which will be updated online according to Algorithm \ref{alg:onlineupdate} and Remark \ref{rem:online_updates}.

Since the nested STL formula contains $\vee$ operator, we need to determine which branch out of two branches $ \{\bm{p}_1\}$ and $\{\bm{p}_2, \bm{p}_3\}$ (as in Example \ref{ex:complete paths}) needs to be satisfied. The guideline to choose the branch is as follows: if the initial condition $x_0 \in \mathbb{X}_1 $, we can choose $\Pi_l = \{ f_1,f_2\}$; if $x_0 \in \mathbb{X}_2 $, we can choose $\Pi_l = \{ f_3,f_4, f_5\}$; if $x_0 \notin \mathbb{X}_0$, then the proposed scheme fails to generate a control signal with correctness guarantee and a larger input bound is expected.

It is worth highlighting that in the special case of $\Pi_l = \{ f_1, f_2\}$, the feasibility of QP is guaranteed since for the time domains that $\mathfrak{b}_1$ and $\mathfrak{b}_2$ overlap, they pose the same CBF condition, so only one CBF is active at every time instant. We also observe that, empirically, the feasibility problem can be mitigated by further shrinking the input set $U^\prime$ or enlarging the class $\mathcal{K}$ functions in the QP, for example, by increasing the gain when it is linear. To incorporate heuristics in the control synthesis scheme, in this section, we choose $u_{nom}(t) $ in a way that guides the trajectory towards fulfilling the CBF $\mathfrak{b}_i$ with the smallest $\bar{t}_{\mathfrak{b}_i}$  among all active ones. Several other heuristics are also implemented in the code.

Now we demonstrate the numerical results with the proposed CBF-based QP control synthesis scheme. In Fig. \ref{fig:traj single integrator}, we illustrate  trajectories with time snapshots starting from $(-6,2)$ and $(-2,3.5)$, both of which lie within $\mathbb{X}_1 \cap \mathbb{X}_2$. Here we set the $\alpha_i$ in \eqref{eq:online_syn_lth_group} to be $\alpha_i(v) = v, v\in \mathbb{R}, \forall i$, and $Q$ in \eqref{eq:online_syn_lth_group} to be an identity matrix. For all the trajectories, the input bound $U$ is respected. We observe that every trajectory satisfies the STL specification $\varphi$. If we take the initial condition $x_0 = (-20,-5)$, $x_0\in \mathbb{X}_1$ and $x_0\notin \mathbb{X}_2$, we observe that the STL specification is fulfilled if we choose the branch $\{ \bm{p}_1\}$; yet the online QP becomes infeasible if we choose the branch $\{ \bm{p}_2, \bm{p}_3\}$. This is in line with the theoretical results.

\begin{figure} 
    \centering
    \includegraphics[width=0.9\linewidth]{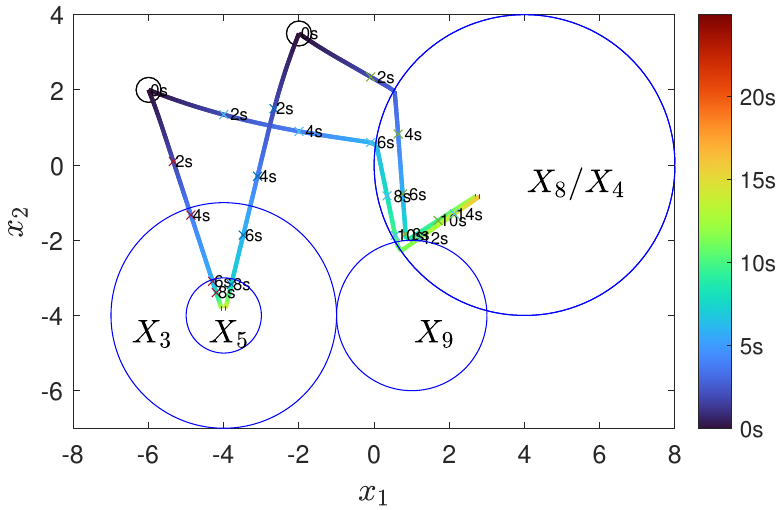}
    \caption{ Four trajectories of a mobile robot with single integrator dynamics are synthesized using the proposed method under the STL specification $\varphi=\mathsf{F}_{[0, 15]}(\mathsf{G}_{[2, 10]}\mu_1 \vee \mu_2  \mathsf{U}_{[5, 10]}\mu_3)$.  Two different starting positions (marked in circles) are tested and every trajectory satisfies the STL specification $\varphi$. It is observed that the robot starts its voyage from $0$s and approaches the regions of interest without any stop. This was enabled by the design of the nominal controller and the online updates of the start time intervals of the set nodes. For the case that branch $\{ {\bm p}_2, {\bm p}_3\}$ is chosen, we note that the trajectories leave $\mathbb{X}_9$ after reaching it. This behavior is due to the approximation gap as we require the trajectories to stay inside $\mathbb{X}_8$ after reaching $\mathbb{X}_9$ for the constructed sTLT. 
    }
    \label{fig:traj single integrator}
\end{figure}

When dealing with regions of irregular shapes or general nonlinear dynamics, the set nodes as well as the CBFs are difficult to calculate analytically. In the following we show a numerical construction scheme.

\subsection{Unicycle model}

Consider a mobile robot with a unicycle dynamics 
\begin{equation} \label{eq:unicycle}
    \begin{aligned}
    \dot{x}_1&=v\cos\theta, \\
\dot{x}_2&=v\sin\theta, \\
\dot{\theta}&=\omega,
    \end{aligned}
\end{equation}
where the state $x = (x_1,x_2, \theta)$,  the control input $u=(v,\omega)$. Here $(x_1,x_2)$ denotes the position, $\theta$ the heading angle, and $v$ the velocity, $\omega$ the turning rate. We assume that the control input $u=(v,\omega) \in U = \{ u\mid |v|\le 1, |\omega|\leq 1\}$. The STL task specification is again given by $\varphi=\mathsf{F}_{[0, 15]}(\mathsf{G}_{[2, 10]}\mu_1 \vee \mu_2  \mathsf{U}_{[5, 10]}\mu_3)$ (the same as in Example \ref{example1}), where $\mathbb{S}_{\mu_1}=\{x\in \mathbb{R}^2 \times S^1\mid  (x_1+4)^2+(x_2+4)^2\leq 1\}$,  $\mathbb{S}_{\mu_2}=\{x\in \mathbb{R}^2 \times S^1\mid  (x_1-4)^2+x_2^2\leq 4^2\}$, and $\mathbb{S}_{\mu_3}=\{x\in \mathbb{R}^2 \times S^1 \mid (x_1-1)^2+(x_2+4)^2\leq  2^2 \}$. Recall from Example \ref{example1}, the sTLT $\mathcal{T}_{\hat\varphi}$ is plotted in Fig. \ref{Fig:example}.

We note that the temporal fragments, their time encodings,  the time domains for the barrier functions, and the branch choosing guidelines are similar to those as in the case of single integrator dynamics and thus omitted here. We will instead explain how the set nodes as well as the barrier functions are constructed through the use of a level-set reachability analysis toolbox \cite{mitchell2005toolbox,fisac2015reach}. 

Here the  set nodes with the  input set $ U = \{ u\mid |v|\le 1, |\omega|\leq 1\}$  are computed via reachability analysis. We may also use a shrinked input set to mitigate the online QP infeasibility issue. In brevity, we numerically obtain the value function $h_{\mathbb{X}_i}$ for the sets $\mathbb{X}_i, i = 0,1,..,9,$ following the reachability operations in Example \ref{example1}.  The projection of sets $\mathbb{X}_3, \mathbb{X}_4, \mathbb{X}_5,\mathbb{X}_8,\mathbb{X}_9$ to the first two dimensions are depicted in Fig. \ref{fig:traj unicycle}. 

\begin{figure} 
    \centering
    \includegraphics[width=0.9\linewidth]{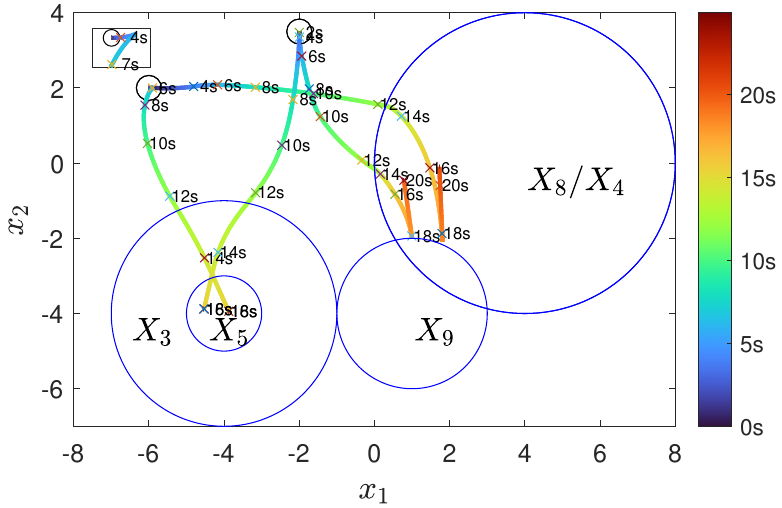}
    \caption{Four trajectories of a mobile robot with unicycle dynamics are synthesised using the proposed method under the STL specification $\varphi=\mathsf{F}_{[0, 15]}(\mathsf{G}_{[2, 10]}\mu_1 \vee \mu_2  \mathsf{U}_{[5, 10]}\mu_3)$.  Two different starting configurations (marked in circles) are tested and every trajectory satisfies the STL specification $\varphi$. It is observed that the robot first adjusts its orientation and then starts to approach regions of interest without any stop. For the case that branch $\{ {\bm p}_2, {\bm p}_3\}$ is chosen, we also see that the trajectories leave $\mathbb{X}_9$ after reaching it. }

    \label{fig:traj unicycle}
\end{figure}

Now we show how the CBFs are constructed. Take the construction of $\mathfrak{b}_2$ as an example, which corresponds to  $f_2 = \mathsf{G}_{[2, 10]} \mathbb{X}_{5}$. Recall $\mathbb{X}_{5} = \{x\mid h_{\mathbb{X}_5}(x) \ge 0 \},$ $[\underline{t}_s(\mathbb{X}_{5}), \bar{t}_s(\mathbb{X}_{5})]=[2, 17]$, $[\underline{t}_{\mathfrak{b}_2},\bar{t}_{\mathfrak{{b}}_2}] = [2, 25]$.  Here the function $\mathfrak{b}_2$ is expected to be  a valid control barrier function for the unicycle dynamics in \eqref{eq:unicycle} which guides $\bm{x}(t)$ towards the set $\mathbb{X}_5$ for $t\in [2,17]$ and keeps $\bm{x}(t)$ in the set  $\mathbb{X}_5$ for $t\in [17,25]$. We construct such a $\mathfrak{b}_2$ by solving the following optimal control problem:
\begin{equation}
    \begin{aligned}
     & V(x,t) = \max_{{\bm u}(s), s \in [t,17]} h_{\mathbb{X}_5}({\bm x}_{x,t}^{\bm u} (17)) \\
     & \text{ s.t. } \hspace{0.5cm} \eqref{eq:unicycle} \text{ and } {\bm u}(s) \in U,
    \end{aligned}
\end{equation}
where $\bm{x}_{x,t}^{\bm u}$ denotes the continuous state signal starting from $x$ at time $t$ with the input signal ${\bm u}$. $V(x,t)$ can be computed  numerically by solving the following Hamilton-Jacobi-Bellman (HJB) equation \footnote{ In general, if $h_{\mathbb{X}_5}(x)$ is Lipschitz continuous but not smooth, then only the viscosity solution can be obtained from the HJB equation and in this case $V(x,t)$ is Lipschitz continuous, which is differentiable almost everywhere. }
\begin{equation*}
    \begin{aligned}
    & \frac{\partial V}{ \partial t} + \max_{u\in U}  \langle \nabla_x V(x,t), f(x,u) \rangle  = 0, \\
    & V(x,17) = h_{\mathbb{X}_5}(x).
    \end{aligned}
\end{equation*}
Thus, we choose $\mathfrak{b}_2(x,t)= \begin{cases} 
V(x,t), & t\in [15,17]; \\
h_{\mathbb{X}_5}(x), & t \in [17,25].
\end{cases}$ One can verify that $\mathfrak{b}_2(x,t)$ is a valid CBF as per Definition \ref{def:cbf}. The remaining barrier functions $\mathfrak{b}_1, \mathfrak{b}_3, \mathfrak{b}_4, \mathfrak{b}_5$ are constructed in a similar manner.

The numerical results with the proposed scheme for unicycle dynamics are shown in Fig. \ref{fig:traj unicycle}. Here we illustrate the trajectories with time snapshots starting from $(-6,2,0)$ and $(-2,3.5,\pi/2)$, both of which lie within $\mathbb{X}_1 \cap \mathbb{X}_2$. Here the $\alpha_i$ in \eqref{eq:online_syn_lth_group} is set to be $\alpha_i(v) = v, v\in \mathbb{R}, \forall i$, and $Q$ in \eqref{eq:online_syn_lth_group} an identity matrix. An intuitive nominal controller similar to the single integrator case is also utilized in this example. For all the trajectories, the input bound $U$ is respected. Again, we observe that every trajectory satisfies the STL specification $\varphi$.

\subsection{Examples for more complex specifications}

In this subsection, we consider the more complex STL formula below
\begin{multline}\label{eq:complex_STL}
    \varphi = \mathsf{G}_{[0,1]} \mathsf{F}_{[2,3]} \mu_1 \wedge \mathsf{F}_{[6,7]}\mathsf{G}_{[1,2]} \mu_2 
    \wedge \mathsf{F}_{[13,14]} (\mu_3 \mathsf{U}_{[1,4]} \mu_1) \\ \wedge \mathsf{G}_{[0,20]} \neg \mu_4 \wedge \mathsf{F}_{[15,20]}\mu_5 .
\end{multline}

\begin{figure}[t]
    \centering
    \includegraphics[width = 0.9\linewidth]{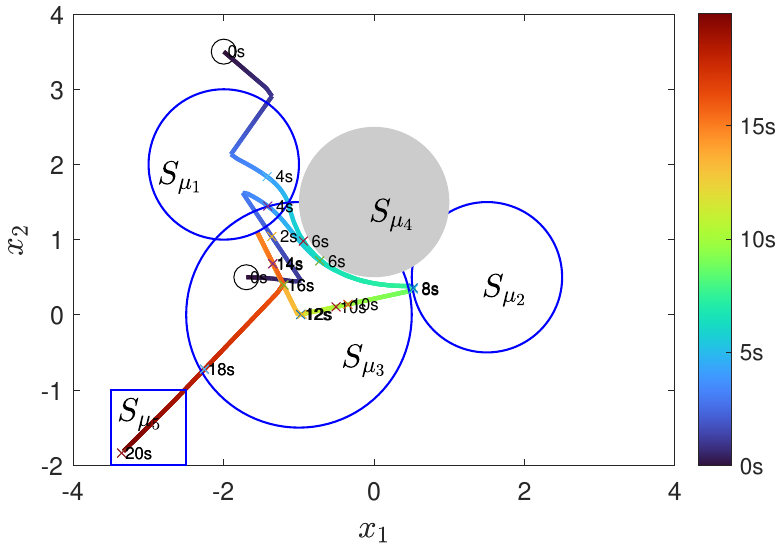}
    \caption{ Synthesised trajectories of a mobile robot with single integrator dynamics in \eqref{eq:single integrator} under the STL task in \eqref{eq:complex_STL}. Here $S_{\mu_1}, S_{\mu_2}, S_{\mu_3}$, $S_{\mu_4}$, and $S_{\mu_5}$ represent the regions in which  the corresponding predicate functions are evaluated to be true. Two different starting positions (marked in circles) are tested and shown in the figure. For both trajectories, it is observed that the robot successfully follows the STL specifications and visits regions of interest in the specified time intervals, and always avoids the obstacle region.} Thus, both trajectories satisfy the STL specification $\varphi$.
    \label{fig:traj_single_integrator_complex}
\end{figure}

\begin{figure}
    \centering
    \includegraphics[width = 0.9\linewidth]{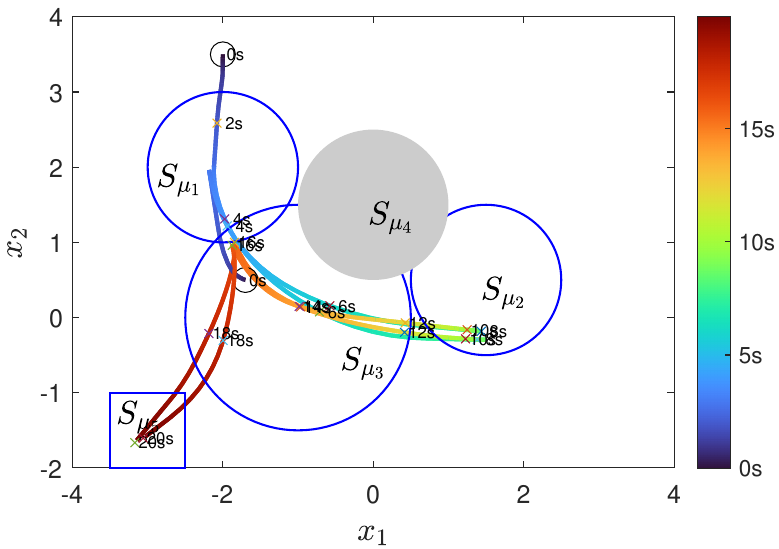}
    \caption{ Synthesised trajectories of a mobile robot with unicycle dynamics in \eqref{eq:unicycle} under the STL task in \eqref{eq:complex_STL}. Here $S_{\mu_1}, S_{\mu_2}, S_{\mu_3}$ and $S_{\mu_4}$ represent the projected regions in which the corresponding predicate functions are evaluated to be true. Two different starting configurations (marked with car-like symbols) are tested and shown in the figure. For both trajectories, it is observed that robot first adjusts its orientation and goes into regions of interest in the specified time intervals, and always avoids the region $S_{\mu_4}$.  Thus, both trajectories satisfy the STL specification $\varphi$. } 
    \label{fig:traj_unicycle_complex}
\end{figure}

The control synthesis process consists of constructing the corresponding sTLT, calculating the set nodes using reachability analysis, calculating the time encodings, and constructing the corresponding CBFs. This offline design process is similar to what we have detailed before, except that the region associated with the predicate $\mu_5$ is a square. We take two different approaches: in the case of single integrator dynamics (Fig. \ref{fig:traj_single_integrator_complex}), we use the largest inscribed circular region to under-approximate $\mathbb{S}_{\mu_5}$, and analytical CBFs are constructed; in the case of unicycle dynamics (Fig. \ref{fig:traj_unicycle_complex}), we use the signed distance function of the square as the superlevel set function and calculate the value function to the HJB equation as the CBF. Implementation details can be found in the online code repository. For the online synthesis, since the formula does not contain $\vee$, the QP in \eqref{eq:online_control_wo_disjunction} will be used.   Figure \ref{fig:traj_single_integrator_complex} and Figure \ref{fig:traj_unicycle_complex} demonstrate the resulting system behaviors for a mobile robot with single integrator dynamics in \eqref{eq:single integrator} and with the unicycle dynamics in \eqref{eq:unicycle}, respectively. We note that all  trajectories satisfy the STL specification in (\ref{eq:complex_STL}), while respecting the dynamics and input bounds.

\section{Conclusions}
In this paper, we develop an efficient control synthesis approach for continuous-time dynamical systems under nested STL specifications. To this purpose, we introduce a notion of signal temporal logic tree (sTLT), detail on its construction from a given STL formula, its semantics (i.e., satisfaction condition), and the equivalence or under-approximation relation between the sTLT and the STL formula. Under the guidance of the sTLT, we show how to design CBFs and online update their activation time intervals. The control signal is thus given by an online CBF-based program. For future work, we will tackle the motion coordination problem of multi-agent systems under STL specifications leveraging task decomposition and distributed CBF techniques.

\bibliographystyle{IEEEtran}
\bibliography{references}

\end{document}